\documentclass[12pt]{article}
\usepackage{amsmath}
\usepackage{graphicx}
\usepackage{enumerate}
\usepackage{natbib}
\usepackage{url} 

\usepackage{authblk}
\usepackage{subfigure}
\usepackage{amssymb,amsthm}
\usepackage{wrapfig}
\usepackage{multirow}
\let\chapter\section 
\usepackage[titlenumbered,linesnumbered,ruled,vlined]{algorithm2e}
\usepackage{caption}
\usepackage{bm}
\usepackage{booktabs}
\usepackage[smallscripts]{moresize}

\usepackage{textcomp} 

\newcommand{\blind}{1}

\usepackage{hyperref}
\hypersetup{ 
    colorlinks,
    citecolor=blue,
    filecolor=blue,
    linkcolor=blue,
    urlcolor=blue 
}

\newtheorem{lemma}{Lemma}
\newtheorem{proposition}{Proposition}
\newtheorem{theorem}{Theorem}

\newtheorem{corollary}{Corollary}
\newtheorem{condition}{Condition}

\newcommand{\bmm}{\mbox{\bf m}}

\newcommand{\be}{\mbox{\bf e}}

\newcommand{\bu}{\mathbf{u}}
\newcommand{\bv}{\mathbf{v}}
\newcommand{\bw}{\mbox{\bf w}}
\newcommand{\bx}{\mbox{\bf x}}
\newcommand{\by}{\mbox{\bf y}}

\newcommand{\bC}{\mbox{\bf C}}
\newcommand{\bD}{\mbox{\bf D}}
\newcommand{\bE}{\mbox{\bf E}}

\newcommand{\bI}{\mbox{\bf I}}

\newcommand{\bP}{\mbox{\bf P}}

\newcommand{\bS}{\mbox{\bf S}}
\newcommand{\bU}{\mbox{\bf U}}
\newcommand{\bV}{\mbox{\bf V}}

\newcommand{\bX}{\mbox{\bf X}}

\newcommand{\bY}{\mbox{\bf Y}}

\newcommand{\bZ}{\mbox{\bf Z}}

\newcommand{\bzero}{\mbox{\bf 0}}
\newcommand{\bveps}{\mbox{\boldmath $\varepsilon$}}

\newcommand{\bepsilon}{\mbox{\boldmath $\epsilon$}}

\begin{document}

%
%


\def\spacingset#1{\renewcommand{\baselinestretch}%
{#1}\small\normalsize} \spacingset{1}


\if1\blind
{
  \title{\bf Sequential scaled sparse factor regression}
  \author{Zemin Zheng, Yang Li, Jie Wu and Yuchen Wang 
  	 \hspace{.2cm}\\
    International Institute of Finance, The School of Management, University of Science and Technology of China} 
  \maketitle
} \fi

\if0\blind
{
  \bigskip
  \bigskip
  \bigskip
  \begin{center}
    {\LARGE\bf Sequential scaled sparse factor regression}
\end{center}
  \medskip
} \fi

\bigskip
\begin{abstract}
 Large-scale association analysis between multivariate responses and predictors is of great practical importance, as exemplified by modern business applications including social media marketing and crisis management. Despite the rapid methodological advances, how to obtain scalable estimators with free tuning of the regularization parameters remains unclear under general noise covariance structures. In this paper, we develop a new methodology called sequential scaled sparse factor regression (SESS) based on a new viewpoint that the problem of recovering a jointly low-rank and sparse regression coefficient matrix can be decomposed into several univariate response sparse regressions through regular eigenvalue decomposition. It combines the strengths of sequential estimation and scaled sparse regression, thus sharing the scalability and the tuning free property for sparsity parameters inherited from the two approaches. The stepwise convex formulation, sequential factor regression framework, and tuning insensitiveness make SESS highly scalable for big data applications. Comprehensive theoretical justifications with new insights into high-dimensional multi-response regressions are also provided. We demonstrate the scalability and effectiveness of the proposed method by simulation studies and stock short interest data analysis.
\end{abstract} 

\noindent%

\spacingset{1.45} 
\section{Introduction} \label{Sec1}
Highly developed technologies and devices have brought in massive data sets in various fields ranging from health and bioinformatics to marketing and economics. In these big data applications, modeling complex dependence structures of multivariate outcomes using observed features is of great importance since it reveals domain knowledge behind the data. For instance, inferring the influence networks from user activities has wide applications in social media marketing \citep{gomez2010inferring} and crisis management \citep{starbird2012will}. By representing the dependency between the outcomes and the predictors through a jointly low-rank and sparse structure, thus alleviating the curse of dimensionality and facilitating model interpretability, sparse reduced-rank regression has gained increasing popularity in large-scale association analyses.

Depending on how the regression coefficient matrix is recovered, sparse reduced-rank regression can generally be grouped into two classes. One is directly estimating the regression coefficient matrix via different kinds of regularization \citep{Anderson1951, Izenman1975, yuan2007dimension, bunea2011optimal, candes2009tight, giraud2011low, negahban2011estimation, bunea2012joint, chen2012, chen2013, lian2015, liu2015calibrated, Goh2017,Fan2019}, where $L_1$ and nuclear norm penalizations are convex relaxations popularly employed to enforce sparse and low-rank structures, respectively. The other is to recover the coefficient matrix from a latent factor point of view by combining the estimated sparse singular vectors based on singular value decomposition \citep{chen2012reduced, mishra2017sequential, Uematsu17, bahadori2016scalable}. Compared with the former class, this type of methods generally enjoy lower computational cost and can be efficiently parallelized in various computing devices. In particular, the sequential estimation procedures proposed in \cite{mishra2017sequential} and \cite{bahadori2016scalable} demonstrate scalability in large-scale applications by decomposing the estimation of the entire coefficient matrix into unit rank matrix recovery problems. They are guaranteed to stop in a few steps under low-rank structures. Nevertheless, these sequential approaches need to tune the optimal sparsity parameter in each step, which can vary between different layers and account for major computational cost when tuned over a wide range of potential values by either cross-validation or some information criterion. 



To further enhance the scalability, it is of urgent need to develop methodology which enjoys free tuning of the regularization parameters for large-scale association analyses. For high-dimensional univariate response sparse regression, tuning free methods have been proposed in \cite{belloni2011square} and \cite{sun2012scaled}, where the universal regularization parameter controlling the sparsity was shown to be independent of the noise level. It was achieved through obtaining an equilibrium between iteratively estimating the noise level via the mean residual square and scaling the penalty in proportion to the estimated noise level. However, it is much more difficult to develop tuning insensitive methods for large-scale multi-response regression since the population covariance matrix of the noise vector can adopt general high-dimensional structures and its sample estimate is usually not invertible to formulate a joint estimation procedure. When the noise covariance matrix is diagonal, meaning that the noises related to different responses are uncorrelated, \cite{liu2015calibrated} proposed the calibrated multivariate regression to attain tuning insensitiveness for either nuclear norm or sparsity penalization. For general structures of the noise covariance matrix or jointly low-rank and sparse coefficient matrix, to the best of our knowledge, there is no existing work which enjoys tuning free property for the sparsity parameters as estimating an invertible noise covariance matrix needs extra penalization in high dimensions.  

In this article, we develop a new methodology for high-dimensional multi-response regression called sequential scaled sparse factor regression (SESS), which combines the strengths of sequential estimation and scaled sparse regression, thus sharing the scalability and the tuning free property for sparsity parameters inherited from the two approaches. The main contributions of this paper are as follows. First of all, we rigorously prove that the problem of recovering a jointly low-rank and sparse regression coefficient matrix can be decomposed into several univariate response sparse regressions through regular eigenvalue decomposition, which provides a new viewpoint on large-scale association analyses. Compared with the sparse eigenvalue problem, regular eigenvalue decomposition is convex and guaranteed to converge. Second, based on the new viewpoint, the proposed approach SESS adopts a universal sparsity parameter in the subsequent univariate response regressions, which is among the first attempts to achieve tuning insensitive jointly low-rank and sparse estimation under general noise covariance structures. Accompanied with a simple BIC-type information criterion for identifying the true rank, whose choices are discrete and thus demonstrating significant gaps between the correct one and the other candidates, SESS is tuning insensitive in both sparsity and rank. The stepwise convex formulation, sequential factor regression framework, and tuning insensitiveness make SESS highly scalable for big data applications. Last but not least, we provide comprehensive theoretical justifications on the effectiveness of the suggested methodology including consistency in estimation, prediction, and rank selection under mild and interpretable conditions, which reveal new insights into high-dimensional multi-response regression. 

The rest of the paper is organized as follows. Section \ref{sec2} presents the model setting and our new methodology. We establish asymptotic properties of the proposed method in Section \ref{sec3}. In Section \ref{sec4}, we verify the theoretical results empirically through simulation examples. An application to the stock short interest data is provided in Section \ref{real data}. Section \ref{sec6} concludes with extensions and possible future work. All technical details are relegated to the Supplementary Material.

\section{Sequential scaled sparse factor regression} \label{sec2} 

\subsection{Model setting}\label{sec2.1}

Consider the following multi-response regression model in the fixed design setting
\begin{equation}\label{mod1}
\bY = \bX \bC^* + \bE,
\end{equation}
where $\bY = (\by_{1},...,\by_{n})^{T}$ denotes an $n\times q$ matrix with $q$ responses, $\bX = (\bx_{1},...,\bx_{n})^{T}$ is an $n\times p$ design matrix with $p$ predictors, $\bC^*$ is an unknown $p \times q$ coefficient matrix, and $\bE = (\be_{1},...,\be_{n})^{T}$ is an $n\times q$ random error matrix with each row vector $\be_i$ independent and identically distributed (i.i.d.) as $N(\mathbf{0}, \mathbf{\Sigma})$\footnote{The Gaussian assumption is not essential and we will show the validity of the proposed method under sub-Gaussian errors in Section \ref{sec3}.}. The columns of $\bX$ are standardized to have a common $L_2$-norm $\sqrt{n}$. Both dimensions $p$ and $q$ are allowed to diverge non-polynomially with the sample size $n$ and $\bC^*$ is assumed to be jointly low-rank and sparse (in rows), entailing the selection of significant predictors. 

Similar to \cite{mishra2017sequential} and \cite{bahadori2016scalable}, we will recover the coefficient matrix $\bC^*$ from a latent factor point of view to facilitate sequential estimation. Specifically, based on the SVD representation of $n^{-1/2}\bX\bC^*$, we have
\begin{align}
&\frac{1}{\sqrt{n}}\bX \bC^* = (\frac{1}{\sqrt{n}} \bX\bU^*) \bD^*_0 \bV_0^{*T}, \nonumber \label{eq:C*1} \\
&\text{s.t.} \quad (\frac{1}{\sqrt{n}}\bX\bU^{*})^T (\frac{1}{\sqrt{n}}\bX \bU^{*}) = \bV_0^{*T} \bV_0^{*} = \bI_{r^*},  
\end{align}
where $\bC^* = \bU^* \bD^*_0 \bV_0^{*T} \in \mathbb{R}^{p \times q}$, $\bU^* = (\bu_1^*, \dots, \bu_{r^*}^*) \in \mathbb{R}^{p \times r^*}$, $\bD_0^* = \text{diag}\{d_1^*, \dots, d_{r^*}^*\} \in \mathbb{R}^{r^* \times r^*}$, $\bV^*_0 \in \mathbb{R}^{q \times r^*}$, and $r^*$ is the rank of $\bC^*$ that allows to be divergent. In the population level, $\bV^{*}_0$ can be identified as the right singular vectors of $\bX\bC^*$. Then $\bU^*$ can be obtained from $\bC^*\bV^{*}_0$ after rescaling the columns even if $p$ is larger than $n$, which gives the feasibility of the above decomposition.

It is worth pointing out that the nonzero singular values in $\bD^*_0$ can diverge with the dimensionality $q$ and their magnitudes can be as large as $n^{-1/2}\|\bX \bC^*\|_F$, which is around the order of $\sqrt{q}$ when each component of the $n$ by $q$ noiseless response matrix $\bY^* = \bX \bC^*$ is around a constant level. To ensure the identifiability of the left singular vectors $\bu_k^*$, we assume that $\bu^*_k$ are sparse (inherited from the row sparsity of $\bC^*$) and $\bu^*_k \perp \text{Ker}(\bP)$ with $\text{Ker}(\bP)$ the null space of the Gram matrix $\bP = n^{-1} \bX^T \bX$. For the right singular vectors, we do not impose any sparsity constraint and will discuss in Section \ref{Sec2.2} how to obtain sparse estimates if the population ones are indeed sparse.




For ease of presentation, we set $\bV^* = \bV_0^{*} \bD_0^* = (\bv_1^*, \dots, \bv_{r^*}^*) \in \mathbb{R}^{q \times r^*}$ so that the right singular vectors absorb the singular values and are no longer of unit length. It yields
\begin{align}
&\bC^* = \bU^* \bV^{*T}= \sum_{k=1}^{r^*} \bu_k^* \bv_k^{*T} = \sum_{k=1}^{r^*} \bC_k^*, \label{eq:C*}
\end{align}
where $\bC^*_k = \bu_k^* \bv_k^{*T}$ is the $k$th layer unit rank matrix of $\bC^*$. Here the singular vectors are sorted by the magnitudes of the singular values of $\bX \bC^*$, consistent with the contribution to the prediction of $\bY$. Generally speaking, the decomposition in the form of $\bC^* = \bU^* \bV^{*T}$ is not unique without orthogonality constraints. But decomposition (\ref{eq:C*}) is the special one that gives $r^*$ uncorrelated latent factors $\bX \bu_k^*$ in view of the orthogonality of $n^{-1/2}\bX \bU^*$ in \eqref{eq:C*1}. Each latent factor is a linear combination of a small subset of the predictors due to the sparsity of $\bu^*_k$. Our goal is to scalably and accurately estimate the singular vectors $\bu_k^*$ and $\bv_k^*$, as well as  the true rank $r^*$, so that the latent factors and their impacts can be recovered.

\subsection{Scalable estimation by SESS}\label{Sec2.2}

The proposed method SESS is motivated by the fact that in the noiseless case $\bY^* = \bX \bC^*$ with $\bC^*$ adopting decomposition (\ref{eq:C*}), the latent factors $\bZ_k^* = \bX \bu_k^*$, $1 \leq k \leq r^*$, are the top-$r^*$ eigenvectors of the following eigenvalue problem
\begin{align*}
(nq)^{-1} \bY^*\bY^{*T} \bZ = \lambda \bZ.
\end{align*}
The corresponding eigenvalues
\begin{align}\label{lambda}
\lambda_k = (nq)^{-1} \|\bX \bC^*_k\|_F^2 = q^{-1} d_k^{*2}
\end{align}
are typically around the constant level based on the discussion on the magnitudes of $d_k^{*}$ after (\ref{eq:C*1}). Moreover, it can be verified that the right singular vectors $\bv_1^*, \dots, \bv_r^*$ satisfy the following intrinsic relationship with $\bu_1^*, \dots, \bu_r^*$,
\begin{align}\label{v}
\bv_k^* = \frac{1}{\bu_k^{*T} \bX^T \bX \bu_k^{*}} \bY^{*T} \bX\bu_k^{*} = \frac{1}{n} \bY^{*T} \bX\bu_k^{*}= \frac{1}{n} \bY^{*T} \bZ_k^{*}.
\end{align}

Therefore, with data matrix $(\bX, \bY)$, we propose to recover the left singular vectors $\bu_k^*$ sequentially in two steps. The first step is to solve the regular eigenvalue problem
\begin{align}\label{step1}
(nq)^{-1}  \bY\bY^{T}\bZ & = \lambda \bZ
\end{align}
and get the estimated latent factors $\widehat{\bZ}_k$ with $\|\widehat{\bZ}_k\|_2 = \sqrt{n}$ as well as the corresponding eigenvalues $\widehat{\lambda}_k$. Then various kinds of regularization methods can be applied to recover the singular vectors $\bu_k^{*}$. See \cite{tibshirani1996regression, FanLi2001, Zou2006, Candes2007, Fan2009, FanLv2011, FanFan2014, FanyLv2014, YuFeng2014, Weng2019}, among many others.

To facilitate the theoretical analysis, here we utilize the popularly used Lasso \citep{tibshirani1996regression} to obtain the sparse left singular vectors $\widehat{\bu}_k$ by solving 
\begin{align*}
\widehat{\bu}_{k} = {\rm arg}\min_{\bu} \Big\{ \frac{\|\widehat{\bZ}_{k}-\bX\bu\|^{2}_{2}}{2n} + \omega_k \|\bu\|_1 \Big\},
\end{align*}
where $\omega_k$ is a regularization parameter controlling sparsity and needs to be tuned by cross-validation or certain information criterion.
In big data applications, we can further save the tuning of the regularization parameter through the following scaled version of the Lasso \citep{sun2012scaled}
\begin{align*} 
  (\widehat{\bu}_{k}, \widehat{\sigma}_k) & = {\rm arg}\min_{\bu,\sigma} \Big\{\frac{\|\widehat{\bZ}_{k}-\bX\bu\|^{2}_{2}}{2n\sigma}+\frac{\sigma}{2}+\omega_{0}\|\bu\|_{1} \Big\}, 
\end{align*}
where $\omega_0$ is a universal regularization parameter to be specified later. Note that $\widehat{\sigma}_k$ here is no longer an estimate of the error standard deviation but utilized to adjust for the standard error $\|\widehat{\bZ}_{k}-\bX\bu_k^*\|_2/\sqrt{n}$ in the $k$th layer. 
The solution $\widehat{\bu}_{k}$ of the scaled Lasso is the same as that of the Lasso with $\omega_k = \widehat{\sigma}_k \omega_{0}$.

On the other hand, motivated by (\ref{v}), after getting $\widehat{\bZ}_{k}$, the right singular vectors $\bv_{k}^*$ can be estimated as
\begin{align}
&\widehat{\bv}_k = n^{-1} \bY^T \widehat{\bZ}_k. \label{eq:solu} 
\end{align}
We will show in Section \ref{sec3} that the convergence rates of $\widehat{\bv}_k$ are basically the same as that of $\widehat{\bZ}_k$ after adjusting for the corresponding scales. Then based on the estimation consistency, a simple entry-wise thresholding will yield sparse right singular vectors to facilitate the selection of response variables if $\bv^*_k$ are indeed sparse. However, since the sparse structure of $\bv^*_k$ can reduce the magnitude of the singular value $d_k^{*}$ in view of $d_k^{*} = n^{-1/2} \|\bX \bC^*_k\|_F$, if $d_k^{*}$ becomes not that large compared with the noise $\|\bE\|_2$, better accuracy would be achieved by directly estimating the coefficient matrix via some co-sparsity inducing penalty \citep{mishra2017sequential,Uematsu17}.

%

Finally, since the true rank $r^*$ is unknown in practice, we will estimate $\widehat{\bZ}_k$ sequentially until the $k$th eigenvalue $\widehat{\lambda}_k$ of (\ref{step1}) is no larger than certain tolerance level $\mu$, which controls the maximum rank $r$. A simple tuning procedure based on $\widehat{\bZ}_k$ will be provided in Section \ref{sec3} to identify the optimal rank $\widehat{r}$. Then we have unit rank matrices $\widehat{\bC}_k = \widehat{\bu}_k \widehat{\bv}^T_k$ as the estimates of $\bC^*_k$ and the estimated regression coefficient matrix $\widehat{\bC}$ is defined as 
\begin{equation*}
  \widehat{\bC} =\sum_{k = 1}^{\widehat{r}} \widehat{\bC}_k = \sum_{k=1}^{\widehat{r}}\widehat{\bu}_k \widehat{\bv}_k^{T}.
\end{equation*}
The implementation of SESS is summarized in Algorithm \ref{Algo}.

\begin{algorithm}[htbp]
Algorithm SESS. \label{Algo}
\begin{tabbing}
   \qquad \enspace Input: $\bY \in \mathbb{R}^{n\times q}$, $\bX \in \mathbb{R}^{n\times p}$, and termination parameter $\mu$\\
   \qquad \enspace $j \gets 1, \ \widehat{\bY}_{j} \gets 0$\\
   \qquad \enspace repeat\\
   \qquad \qquad $(\widehat{\bZ}_{j}, \widehat{\lambda}_j) \gets j$th eigenvector and eigenvalue of $(nq)^{-1}\bY\bY^T\bZ = \lambda \bZ$ \label{line:eig}\\ 
   \qquad \qquad if $\widehat{\lambda}_j  > \mu$ then \\
   \qquad \qquad \qquad $\widehat{\bv}_j \gets n^{-1} \bY^T \widehat{\bZ}_j$  \\
   \qquad \qquad \qquad $\widehat{\bY}_{j}\gets \widehat{\bY}_{j} + \widehat{\bZ}_j \widehat{\bv}_j^T$ \\
   \qquad \qquad \qquad $j \gets j + 1$\\
   \qquad \qquad end\\
   \qquad \enspace tune the optimal rank $\widehat{r}$ by information criterion (\ref{rank}) \\
   \qquad \enspace $k \gets 1$ \\
   \qquad \enspace repeat\\
   \qquad \qquad $(\widehat{\bu}_{k}, \widehat{\sigma}_k) \gets {\rm arg}\min_{\bu,\sigma}\{(2n\sigma)^{-1} \|\widehat{\bZ}_{k}-\bX\bu\|^{2}_{2} + 2^{-1} \sigma + \omega_{0}\|\bu\|_{1}\}$ \\
   \qquad \qquad $\widehat{\bC}_k \gets \widehat{\bu}_{k}\widehat{\bv}_{k}^{T}$\\
   \qquad \qquad $k \gets k + 1$\\
   \qquad \qquad until $k > \widehat{r}$\\

\qquad \enspace output $\widehat{\bC} = \sum_{k = 1}^{\widehat{r}} \widehat{\bC}_k$
\end{tabbing}
\end{algorithm}

Although SESS adopts a similar sequential estimation framework as those in \cite{mishra2017sequential} and \cite{bahadori2016scalable}, it enjoys two significant advantages. First of all, the optimal sparsity parameter depends on the noise level and varies from layer to layer in \cite{mishra2017sequential} and \cite{bahadori2016scalable}. By contrast, SESS converts the multi-response regression problem into several univariate response regressions, thus making it possible to adopt a universal sparsity parameter $\omega_0$ and substantially improve the computational efficiency. Second, the major sparse estimation technique such as sparse eigenvalue decomposition utilized in \cite{bahadori2016scalable} is a nonconvex optimization problem and may not guarantee convergence under general settings \citep{Ma2013}. But the regular eigenvalue problem and the scaled Lasso estimation render SESS a stepwise convex formulation with guaranteed numerical stability. 

Nevertheless, we need to address two main issues before claiming the success of SESS. The first one is that when solving the regular eigenvalue problem (\ref{step1}), we lose the space constraint on $\bZ$ as it may not fall into the column space of $\bX$. Then how much price in accuracy we pay to trade for the computational efficiency is unknown. The other issue is that the relationship between the estimated latent factors $\widehat{\bZ}_k$ and their population counterparts $\bX \bu^*_k$ can be different from the standard high-dimensional linear regression. Therefore, whether regularization methods such as the Lasso or the scaled Lasso apply and how to choose the corresponding regularization parameters deserve careful investigation. We will provide comprehensive theoretical guarantees with new insights in the next section. 

\section{Theoretical properties} \label{sec3}
This section will present the theoretical properties of the proposed method SESS. First of all, we list the following technical conditions and discuss their relevance. 

\begin{condition} \label{cond2} For some positive constant $d_{\lambda}$, the top-$r^*$ population eigenvalues $\lambda_k$ in \eqref{lambda} satisfy $\lambda_{k} - \lambda_{k + 1} \geq d_{\lambda}$, $k = 1, \ldots, r^*$. 
\end{condition}

\begin{condition}\label{cond3} The eigenvalues of the population covariance matrix $\mathbf{\Sigma}$ for the random error vector are bounded from above and below by positive constants $\gamma_u$ and $\gamma_l$, respectively. 
\end{condition}

\begin{condition} \label{conduv} There exists positive constants $U$ and $V$ such that the lengths of the left and right population singular vectors satisfy $\|\bu^*_k\|_2 \leq U$ and $\|\bv^*_k\|_2/\sqrt{q} \leq V$ for any $k$, $k = 1, \ldots, r^*$.
\end{condition}

Condition \ref{cond2} is imposed to ensure the identifiability of the latent factors with $d_{\lambda}$ the minimum separation between successive nonzero eigenvalues. Similar identifiability assumptions can be found in \citet{Fan2016,Fan2017,Uematsu17}. Condition \ref{cond3} allows for general covariance structure of the random error vector as long as its eigenvalues are bounded. The upper bound controls the noise level while the lower bound is only needed in Theorem \ref{theo2} to facilitate rank selection. It is weaker than the independence assumption on the error vector imposed in \citet{bunea2012joint} and \cite{liu2015calibrated}. Condition \ref{conduv} puts a mild assumption on the lengths of the singular vectors. The constant upper bound on $\|\bu^*_k\|_2$ is reasonable due to the sparsity of $\bu^*_k$ and can be consistent with the aforementioned scale of $\|\bX \bu_k^*\|_2 = \sqrt{n}$ since the design matrix $\bX$ can typically satisfy some sparse eigenvalue assumption \citep{candes2005, Uematsu17, bahadori2016scalable}. 

Besides the above assumptions, we also need to characterize the model identifiability about the sparse left singular vectors by restricting the correlations between the significant predictors and the noise ones. Recall that the Gram matrix $\bP = n^{-1} \bX^T \bX$. Given $\xi \geq 0$ and $S \subset \{1,...,p\}$, the sign-restricted cone invertibility factors introduced in \cite{ye2010rate} are defined as 
\begin{align*}
  F_m (\xi,S) & = \inf \left\{|S|^{1/m}\|\bP \bu\|_{\infty}/ \|\bu\|_{m}: \bu \in \mathcal{G}_{-}(\xi,S) \right\} 
\end{align*}
for positive integer $m$ with the sign-restricted cone $\mathcal{G}_{-}(\xi,S)= \{\bu: \|\bu_{S^c}\|_1 \leq \xi \|\bu_S\|_1 \neq 0, \bu_j \bx_j^T \bX \bu \leq 0, \forall j \notin S\}$. As pointed out in \cite{sun2012scaled}, the bounded sign-restricted cone invertibility factor assumption can be slightly weaker than the parallel condition on the compatibility factor or the restricted eigenvalue \citep{bickel2009simultaneous}. So we put it as follows to ensure the identifiability of $S_k = \rm{supp}$$(\bu_k^*)$, which are the supports of the left population singular vectors.

\begin{condition}\label{cond1}
For certain positive constants $\xi$, $F_1$, and $F_2$, the sign-restricted cone invertibility factors $F_1 (\xi, S_k) \geq F_1$ and $F_2 (\xi, S_k) \geq F_2$ for any $k$, $k = 1, \ldots, r^*$.
\end{condition}




Since both $\widehat{\bZ}_k$ and $-\widehat{\bZ}_k$ are solutions of the regular eigenvalue problem \eqref{step1}, we assume that $\widehat{\bZ}_k$ takes the correct direction to facilitate the theoretical analysis. That is, the angle between $\widehat{\bZ}_k$ and $\bZ^*_k$ is no more than a right angle. Otherwise, we can change $\widehat{\bZ}_k$ to $-\widehat{\bZ}_k$ to satisfy that. Now we are ready to show the main results.

\begin{proposition}[Consistency of latent factors]
\label{ZYYZ}
Under Conditions \ref{cond2}-\ref{cond3}, for sufficiently large $n$, the following statement holds with probability at least $1 - \exp(-\frac{n}{2})$ for any $k$, $1 \leq k \leq r^*$,
\begin{align*}
\frac{1}{\sqrt{n}}\|\widehat{\bZ}_k - \bX\bu^*_k\|_2 \leq \frac{4 \sqrt{\lambda_1} \gamma_u} {d_\lambda} \left(\frac{\sqrt{n} + \sqrt{q}}{\sqrt{nq}} \right).
\end{align*}
\end{proposition}

Proposition \ref{ZYYZ} provides a uniform convergence rate around the order of $(\sqrt{n} + \sqrt{q})/\sqrt{nq}$ for all the latent factors relative to the top-$r^*$ singular values with significant probability. The numerator indicates the magnitude of $\|\bE\|_2$, which is the largest singular value of the error matrix, while the denominator corresponds to the order of the top-$r^*$ singular values $\|\bX \bC_k^*\|_F$. In \citet{bunea2012joint}, it was shown by random matrix theory that $\|\bE\|_2$ is of the magnitude $(\sqrt{n} + \sqrt{q})$ for independent entries and we generalize this result to allow for correlated random errors. Our convergence rate is as fast as $n^{-1/2} + q^{-1/2}$ based on the current signal strength $\sqrt{nq}$ as we do not impose any sparsity assumption on the right singular vectors $\bv_k^*$. In view of the technical argument, the consistency of estimated latent factors can still be guaranteed as long as the signal strengths $\|\bX \bC_k^*\|_F$ are above the noise level. Moreover, the eigenvalue $\lambda_1$ and the eigengap $d_\lambda$ here are assumed to be around the constant level for simplicity under the low rank structure, and the convergence rates still hold even if they become divergent.


The results of Proposition \ref{ZYYZ} are the bases of our two-step procedure for estimating the left singular vectors in SESS, which show that the latent factors can be consistently recovered from a regular eigenvalue problem even if we do not force the solutions to lie in the column space of the design matrix $\bX$. The underlying reason is that the directions achieving the maximum variations remain close to the population ones when the perturbation is relatively small compared with the signals of the factors. Moreover, the penalized regression in the second step is indeed a relaxed projection of $\widehat{\bZ}_k$ to the column space of $\bX$, which further alleviates the issue of lacking subspace constraint in the first step. Then the proposed two-step procedure substantially simplifies the computational complexity compared with the nonconvex generalized sparse eigenvalue problem in \cite{bahadori2016scalable}, making it possible to decompose the original multi-response regression problem into several univariate response regressions. Extra benefits on tuning the sparsity parameter and the rank will be demonstrated through the subsequent theorems.

Note that a random vector $\bw=(w_1,\ldots,w_q)^{T}\in \mathbb{R}^{q}$ is said to be sub-Gaussian distributed if there exists some positive constant $k$ such that the marginal random variable $\bmm^{T}\bw$ satisfies $\mathbb{P}(|\bmm^{T}\bw|>t)\leq\exp(1-t^2/k^2)$ for any $t>0$ and any unit length vector $\bmm\in \mathbb{R}^{q}$. Its second moment matrix is defined as $\mathbb{E}(\bw^{T}\bw)$. Since the Gaussian assumption of the random error vector is not essential for our method, the following corollary generalizes the results of Proposition \ref{ZYYZ} to sub-Gaussian errors. It guarantees that the subsequent theoretical results can also hold for sub-Gaussian errors after some constant adjustment by applying the same technical arguments. 

\begin{corollary}
\label{nds}
Suppose that the rows of the random error matrix $\bE$ in \eqref{mod1} are independent sub-Gaussian vectors with a common second moment matrix $\mathbf{\Sigma}^{\star}$, whose eigenvalues are bounded from above. Then under Condition \ref{cond2}, for sufficiently large $n$ and some constant $\gamma_u'$, it holds that with probability at least $1 - \exp(-\frac{n}{2})$ for any $k$, $1 \leq k \leq r^*$,
\begin{align*}
\frac{1}{\sqrt{n}}\|\widehat{\bZ}_k - \bX\bu^*_k\|_2 \leq \frac{4 \sqrt{\lambda_1} \gamma_u'} {d_\lambda} \left(\frac{\sqrt{n} + \sqrt{q}}{\sqrt{nq}} \right).
\end{align*}
\end{corollary}



With the estimated latent factors $\widehat{\bZ}_k$, our second step is to recover the sparse left singular vectors through penalized regressions. However, this problem can be somewhat between model selection and sparse recovery \citep{candes2005,candes2006,lv2009} since the residual vector converges to zero when regressing $\widehat{\bZ}_k$ on $\bX$ with the true coefficient vector $\bu^*_k$. As the solution $\widehat{\bu}_{k}$ of the Lasso is the same as that of the scaled Lasso with $\omega_k = \widehat{\sigma}_k \omega_{0}$, the following theorem guarantees the estimation accuracy of SESS.


\begin{theorem}[Consistency of sequential estimation]\label{theo1} 
Suppose that Conditions \ref{cond2}-\ref{cond1} hold and $\omega_k = \widetilde{C} (\frac{\sqrt{n} + \sqrt{q}}{\sqrt{nq}}) (\frac{\xi +1}{\xi -1})$ for any constant $\widetilde{C} > \frac{4 \sqrt{\lambda_1} \gamma_u} {d_\lambda}$. Then for sufficiently large $n$, with probability at least $1- \exp(-\frac{n}{2})$, the following statements hold simultaneously for any $k$, $1 \leq k \leq r^*$,
\begin{align*}
& \qquad \ \|\widehat{\bu}_{k} - \bu_{k}^* \|_{2} \leq C_{u}\sqrt{s} \Big(\frac{\sqrt{n} + \sqrt{q}}{\sqrt{nq}}\Big) , \ \frac{1}{\sqrt{n}} \|\bX\widehat{\bu}_{k} - \bX \bu^*_{k}\|_2 \leq \widetilde{C}_{u}\sqrt{s} \Big(\frac{\sqrt{n} + \sqrt{q}}{\sqrt{nq}}\Big), \\
& \quad \frac{1}{\sqrt{q}} \|\widehat{\bv}_{k} - \bv^*_{k}\|_{2} \leq \ C_{v} \Big(\frac{\sqrt{n} + \sqrt{q}}{\sqrt{nq}}\Big) , \ \ \frac{1}{\sqrt{q}} \|\widehat{\bC}_k - \bC^*_k \|_F \leq (V C_{u} + U C_{v}) \sqrt{s} \Big(\frac{\sqrt{n} + \sqrt{q}}{\sqrt{nq}}\Big),  \\
& \|\widehat{\bZ}_{k} \widehat{\bv}_{k}^T - \bX  \bC^*_k \|_F  \leq (V \widetilde{C} + C_{v}) (\sqrt{n} + \sqrt{q}), \ \|\bX \widehat{\bC}_k - \bX  \bC^*_k \|_F  \leq (V \widetilde{C}_{u} + C_{v}) \sqrt{s} (\sqrt{n} + \sqrt{q}),
\end{align*}
where $s = \max_{k = 1}^{r^*} |S_k|$ is the maximum sparsity level, $C_{u} = 2 \widetilde{C} \xi /\{(\xi - 1) F_2\}$, $\widetilde{C}_{u} = 2 \widetilde{C} \xi / \{(\xi - 1) \sqrt{F_1}\}$, and $C_{v} = (4\lambda_1/d_{\lambda} + 2)\gamma_u$ are positive constants.
\end{theorem}

Theorem \ref{theo1} presents a uniform convergence rate of the order $\sqrt{s}(\sqrt{n} + \sqrt{q})/\sqrt{nq}$ for the left singular vectors and the unit rank matrices corresponding to the top-$r^*$ singular values with significant probability. Compared with the uniform convergence rate in Proposition \ref{ZYYZ} and that of the right singular vectors, there is an extra term $\sqrt{s}$ reflecting the price we pay for estimating the nonzero entries in $\bu_{k}^*$. Interestingly, here we do not observe a $\log p$ term that typically exists in high-dimensional regression problems. The $\log p$ term is used to be induced by the penalization parameter of a magnitude no smaller than the maximum spurious correlation $\|n^{-1}\bX^T \bveps\|_{\infty}$ ($\bveps$ denotes a general random error vector) to suppress the noise variables and exclude them from the selected model. By contrast, as the columns of $\bX$ are standardized to have a common $L_2$-norm $\sqrt{n}$, the corresponding maximum spurious correlation in our setup would be
\begin{align}\label{ome}
n^{-1}\|\bX^T (\widehat{\bZ}_k - \bX\bu^*_k)\|_\infty \leq n^{-1/2}\|\widehat{\bZ}_k - \bX\bu^*_k\|_2 = O_{\mathbb{P}}(n^{-1/2} + q^{-1/2}),
\end{align}
which is independent of the dimensionality $p$. Therefore, we can set the penalization level $\omega_k$ according to the convergence rate of $\widehat{\bZ}_k$ established in Proposition \ref{ZYYZ}. It implies that the proposed method can be applicable to arbitrarily high dimensionality as long as the supports of $\bu_{k}^*$ are identifiable (Condition \ref{cond1} holds). Our theoretical results formally justify the numerical performance in Section \ref{sec4}, where both estimation and prediction accuracies maintain around the same level regardless of the increasing dimensionality.

It is worth noticing that under the same high-dimensional multi-response regression setup, the corresponding convergence rate established in \cite{bahadori2016scalable} was shown to be $\sqrt{s\log (pq)/n}$. Therefore, when the signals of factors are relatively strong such that $\sqrt{q} > \sqrt{n/\log (pq)}$, our estimation accuracy can be better since then the required penalization level $\omega_k$ is less than $\|n^{-1}\bX^T \bE\|_{\infty}$, which is around $\sqrt{\log (pq)/n}$. It is also interesting to note that the optimal error rate for estimating $\bX  \bC^*$ in terms of Frobenius norm is $\sqrt{q + s \log (p/s)}$ \citep{bunea2012joint} when considering unit rank matrix estimation with $r^* = 1$, which can be better than our corresponding rate $\sqrt{q} + \sqrt{n}$ in Theorem \ref{theo1} when $q < n$ and around the same order otherwise. It reveals the tradeoff between computational efficiency and estimation accuracy when the signals of factors are not that large. The last error bound in Theorem \ref{theo1} applies to out-of-sample prediction, which contains an extra term $\sqrt{s}$ since then we can only utilize the regression coefficient matrix instead of the estimated latent factor $\widehat{\bZ}_k$. Another advantage of SESS is that it only requires the tolerated sparsity level for model identification in Condition \ref{cond1} be larger than the number of nonzero components in each $\bu^*_k$ instead of that of the whole regression coefficient matrix $\bC^*$, which alleviates the correlation constraints on the design matrix $\bX$.
Furthermore, the specific choice of $\omega_k$ is derived from the requirement that $\omega_k (\xi -1)/(\xi +1) \geq \|\bX^T (\widehat{\bZ}_k - \bX\bu^*_k)\|_\infty/n$, similarly as in \cite{ye2010rate} and \cite{sun2012scaled}.
Based on \eqref{ome}, it suffices to guarantee that
\[ \Big(\frac{\xi -1}{\xi +1}\Big) \omega_k \geq \widetilde{C} \Big(\frac{\sqrt{n} + \sqrt{q}}{\sqrt{nq}}\Big) \geq \frac{\|\widehat{\bZ}_k - \bX \bu^*_k\|_2}{\sqrt{n}}, \]
which yields the choice of $\widetilde{C}$ in Theorem \ref{theo1}.
Moreover, since $\omega_{0} = \omega_k/\widehat{\sigma}_k$ and $\widehat{\sigma}_k = \|\widehat{\bZ}_k - \bX\widehat{\bu}_k\|_2/\sqrt{n}$ should be close to $\|\widehat{\bZ}_k - \bX \bu^*_k\|_2/\sqrt{n}$, we suggest a universal regularization parameter $\omega_{0}$ around the constant level. This is different from the choice of $\sqrt{2(\log p)/n}$ in \citet{sun2012scaled} for model selection. In our numerical studies, setting $\omega_{0} = 1$ gives satisfactory finite sample performance.

Based on the results of Theorem \ref{theo1}, the regression coefficient matrix $\bC^*$ can be accurately recovered once the rank is correctly identified. After decomposing the multi-response regression into univariate response regressions, the optimal rank can be tuned separately from the sparsity parameters in SESS. Moreover, since the true rank corresponds to the underlying number of latent factors, we propose the following BIC-type information criterion based on the estimated latent factors $\widehat{\bZ}_k$ and their factor loadings $\widehat{\bv}_k = n^{-1} \bY^T \widehat{\bZ}_k$. 


\begin{theorem}[Consistency of rank recovery]\label{theo2}
Suppose that Conditions \ref{cond2}-\ref{conduv} hold, ${r^*} (\frac{\log n}{\sqrt{n}})^{1/2} = o(1)$, $r \left( \frac{\sqrt{n} + \sqrt{q}}{\sqrt{nq}} \right)^{1/2} = o(1)$, and $\frac{\sqrt{n}}{\sqrt{q} \log n} = o (1)$. Then for sufficiently large $n$, the following information criterion
\begin{align}\label{rank}
\mathcal{C}({k}) = \sqrt{n} \log \mathcal{L}(k) + k \log n,
\end{align}
where $\mathcal{L}(k) = (nq)^{-1}\|\bY - \widehat{\bY}_{k}\|_F^2$ with $\widehat{\bY}_{k} = \sum_{j=1}^{k} \widehat{\bZ}_j \widehat{\bv}_j^T$, attains its minimum value when $k = r^*$ with probability at least $1 - 2\exp(-\frac{n}{2})$.
\end{theorem}

As pointed out in \citet{Fan2013Tuning parameter}, some power of the logarithmic factor of dimensionality is usually needed in the model complexity penalty to consistently identify the true model in high dimensions. But a BIC-type information criterion still applies here due to the separation of tuning procedures for the rank and the sparsity parameter. Compared with tuning both parameters via a GIC-type information criterion in \cite{bahadori2016scalable}, information criterion (\ref{rank}) enjoys much lower computational cost, as well as better statistical accuracy since the estimation of $\widehat{\bZ}_k$ bypasses the high-dimensional predictors so that their estimation error bounds do not involve $s$ or $\log p$. When the dimensionality $q$ is less than $n$, we can replace $n$ with $q$ in (\ref{rank}) so that the true rank can still be identified with significant probability by a similar technical argument.

\section{Simulation studies} \label{sec4}
In this section, we use simulated data to investigate the finite sample performance of SESS and compare it with four other methods: column-wise Lasso (Lasso), reduced rank regression (RRR), rank constrained group Lasso (RCGL), and sequential co-sparse factor regression (SeCURE). Lasso and RRR are two classical methods which generate sparse and low-rank estimates, respectively. RCGL yields a jointly row-sparse and low-rank estimate that achieves the optimal prediction error rate \citep{bunea2012joint}, while SeCURE sequentially estimates the sparse unit rank matrices to enjoy co-sparse structures in both left and right singular vectors \citep{mishra2017sequential}.

These methods were implemented as follows. Lasso was implemented by R package `lars' with the sparsity parameters tuned by BIC. RRR was implemented using R package `rrpack' with the rank tuned by the criterion of joint rank and row selection (JRRS) proposed in \citet{bunea2012joint}. RCGL was implemented by R package `rrpack' with the sparsity parameter and the rank tuned by JRRS. SeCURE was implemented using R package `secure' and tuned by BIC for both sparsity parameters and the rank. By contrast, the proposed method SESS utilized a universal sparsity parameter $\omega_0 = 1$ and its rank was chosen by the BIC-type information criterion (\ref{rank}). 

There are six performance measures in total for evaluating different methods. The first three measures are: the normalized prediction error (PE) $\|\bY_{\text{test}} - \bX_{\text{test}} \widehat{\bC}\|_F / \|\bY_{\text{test}}\|_F$ based on an independent test sample of size 10000, the normalized estimation error (EE) $\|\widehat{\bC} - \bC\|_F / \|\bC\|_F$, and the rank recovery error (RE) $|\mathrm{rank}(\widehat{\bC}) - \mathrm{rank}(\bC)|$. The fourth and fifth measures are the false positive rate (FPR) and the false negative rate (FNR) suggested in \cite{mishra2017sequential} to evaluate the results of variable selection for different layers, obtained by comparing the sparsity pattern of $(\widehat{\bu}_1,\dots,\widehat{\bu}_{r^*})$ to that of $(\bu^*_1,\dots,\bu^*_{r^*})$. These two measures do not apply to Lasso, RRR, and RCGL since they do not recover the latent factors. The last one is the averaged CPU time for obtaining the corresponding estimate on a PC with 16 GB RAM and Intel Core i7-8700 CPU (3.20 GHz). 

\subsection{Simulation example 1}\label{s1} 
We generated 100 data sets from multivariate regression model (\ref{mod1}) with similar setups as those in \citet{bahadori2016scalable}. For each data set, the rows of $\bX$ were sampled as i.i.d. copies from $N(\mathbf{0},\mathbf{\Sigma}_X)$ with $\mathbf{\Sigma}_X=(0.5^{|i-j|})_{p \times p}$. Similarly, the rows of $\bE$ are i.i.d. from $N(\mathbf{0}, \gamma \mathbf{\Sigma}_E)$ with $\mathbf{\Sigma}_E=(0.5^{|i-j|})_{q \times q}$ and $\gamma = 0.1$. The parameter matrix $\bC^*$ was generated as follows. First, we created a matrix $\bC$ with 90 non-zero entries and each non-zero entry was drawn independently from $N(0,1)$. Second, based on the singular value decomposition of $\bC = \bU\bS\bV^T$, we replaced the first $r$ diagonal components of $\bS$ by $100, 99, \dots ,101 - r$ and others by 0, yielding a jointly sparse and low-rank coefficient matrix $\bC^*$ with around 50 non-zero entries. We considered two different settings with $(n,q,r) = (100,200,3)$ and $(200,300,10)$, respectively. For both settings, the dimensionality $p$ can vary in $\{800,1500,2000\}$. 



\begin{table}
\begin{center} 
\caption{Means and standard errors (in parentheses) of different performance measures in Section \ref{s1}}
\resizebox{450pt}{280pt}{
\begin{tabular}{clccccc}
$p$           &  Method                 &  PE ($\times 10^{-2}$)  & EE ($\times 10^{-2}$)  & RE       & FNR      & FPR           \\
\hline
                    & \multicolumn{6}{c}{$n$ = 100, $q$ = 200, $r$ = 3}                      \\
\hline
                    &Lasso                   & 5.20 (0.01)       & 4.36 (0.02)      & 87.90 (1.12)   & ------
                         & ------        \\
                    &RRR                       & 29.87 (0.01)       & 27.37 (0.01)      & 0 (0)   & ------      & ------       \\
  $800$             &SESS                      & 2.66 (0.00)       & 0.78 (0.01)      & 0 (0)   & 0 (0)      & 0.06 (0.01)   \\
                    &RCGL                      & 2.71 (0.00)       & 1.34 (0.01)      & 0 (0)   & ------          & ------       \\
                    &SeCURE                    & 3.99 (0.02)       & 2.61 (0.08)      & 0 (0)   & 0.13 (0.08)      & 1.85 (0.25)   \\
\hline
                    &Lasso                   & 5.20 (0.01)       & 4.37 (0.02)      & 88.56 (1.25)   & ------       & ------        \\
                    &RRR                      & 28.13 (0.03)       & 26.65 (0.03)      & 0.02 (0.04)   & ------      & ------       \\
   $1500$           &SESS                      & 2.65 (0.00)       & 0.74 (0.00)      & 0 (0)   & 0 (0)      & 0.06 (0.01)   \\
                    &RCGL                      & 2.71 (0.00)       & 1.59 (0.01)      & 0 (0)   &  ------         & ------       \\
                    &SeCURE                    & 3.14 (0.03)       & 2.79 (0.06)      & 0 (0)   & 0 (0)      & 0.21 (0.01)   \\
\hline
                    &Lasso                   & 5.21 (0.01)       & 4.41 (0.02)      & 89.00 (2.97)  & ------       & ------        \\
                    &RRR                       & 28.29 (0.03)       & 27.21 (0.04)      &0.08 (0.02) & ------      & ------       \\
   $2000$           &SESS                      & 2.66 (0.00)       & 0.79 (0.01)      & 0 (0)   & 0 (0)      & 0.05 (0.01)    \\
                    &RCGL                      & 2.74 (0.00)       & 1.25 (0.01)      & 0 (0)   & ------          & ------        \\
                    &SeCURE                    & 3.12 (0.05)       & 3.11 (0.05)     & 0(0)   & 0 (0)      & 0.11 (0.01)    \\
\hline
                    & \multicolumn{6}{c}{$n$ = 200, $q$ = 300, $r$ = 10}                      \\
\hline
                    &Lasso                   & 4.70 (0.01)       & 4.25 (0.02)      & 110.05 (1.23)   & ------       & ------        \\
                    &RRR                      & 25.76 (0.03)       & 27.04 (0.04)      & 1.04 (0.04)   & ------      & ------       \\
  $800$             &SESS                      & 1.87 (0.00)       & 0.87 (0.00)      & 0 (0)    & 0 (0)      & 0.05 (0.01)  \\
                    &RCGL                      & 2.71 (0.00)       & 1.51 (0.00)      & 0 (0)    & ------          & ------      \\
                    &SeCURE                    & 3.22 (0.00)       & 5.25 (0.01)     & 0 (0)    & 0 (0)      & 0.12 (0.01)  \\
\hline
                    &Lasso                   & 4.69 (0.01)       & 4.25 (0.02)      & 123.95 (1.08)   & ------       & ------        \\
                    &RRR                       & 27.84 (0.03)       & 24.36 (0.03)      & 0 (0)   & ------     & ------      \\
  $1500$            &SESS                      & 1.89 (0.00)       & 0.85 (0.00)      & 0 (0)    & 0 (0)      & 0.02 (0.01)  \\
                    &RCGL                      & 2.67 (0.00)       & 1.62 (0.00)      & 0 (0)    & ------          &  ------     \\
                    &SeCURE                    & 4.33 (0.00)       & 6.01 (0.01)     & 0 (0)    & 0 (0)      & 0 (0)  \\
\hline
                    &Lasso                   & 4.70 (0.01)       & 4.23 (0.02)      & 120.06 (1.23)   & ------       & ------        \\
                    &RRR                       & 27.03 (0.03)       & 24.85 (0.03)      & 0 (0)   & ------     & ------       \\
  $2000$            &SESS                      & 1.89 (0.00)       & 0.82 (0.00)      & 0 (0)    & 0 (0)      & 0.02 (0.01)   \\
                    &RCGL                      & 2.70 (0.00)       & 1.54 (0.00)      & 0 (0)    & ------          & ------       \\
                    &SeCURE                    & 5.01 (0.00)       & 6.76 (0.01)     & 0 (0)    & 0.01 (0.01)      & 0.02 (0.01)   \\
\hline

\end{tabular}

}

\label{simu1}
\end{center}

\end{table}

\begin{figure} 
  \centering
  \subfigure[r = 3]{
    \label{label-a} 
    \includegraphics[width=2.7in]{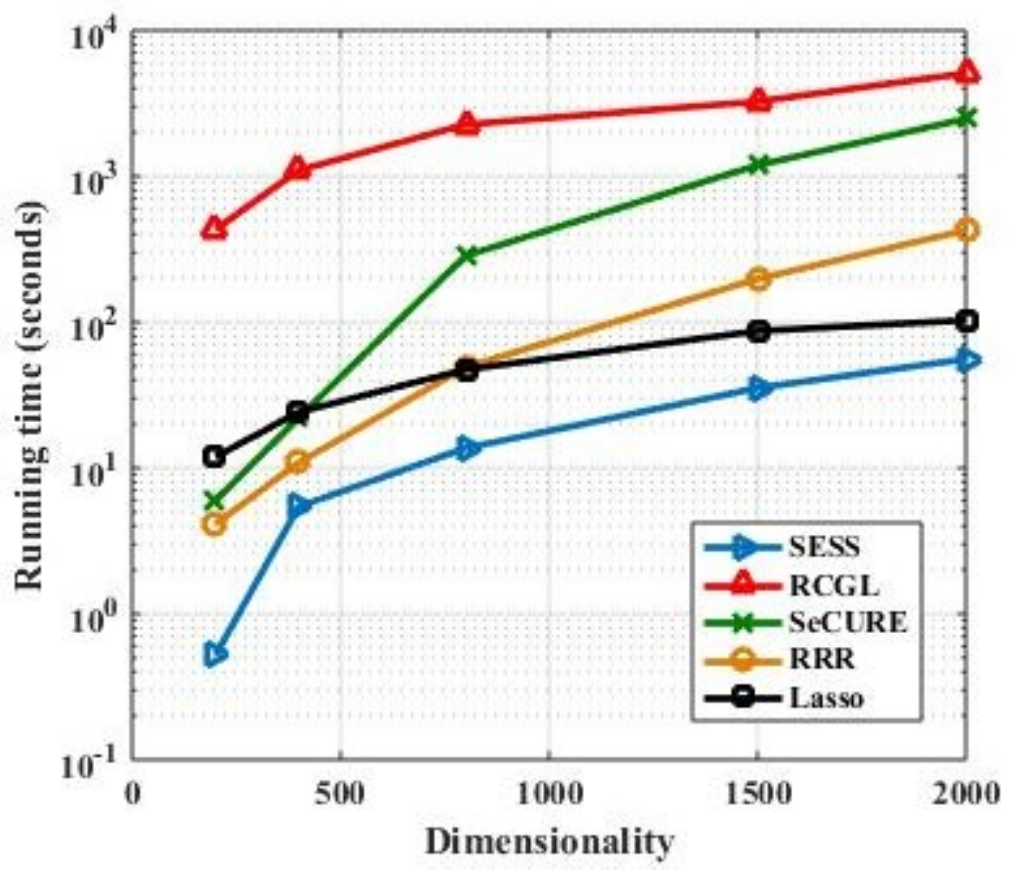}
  }
  \subfigure[r = 10]{
    \label{label-b} 
    \includegraphics[width=2.7in]{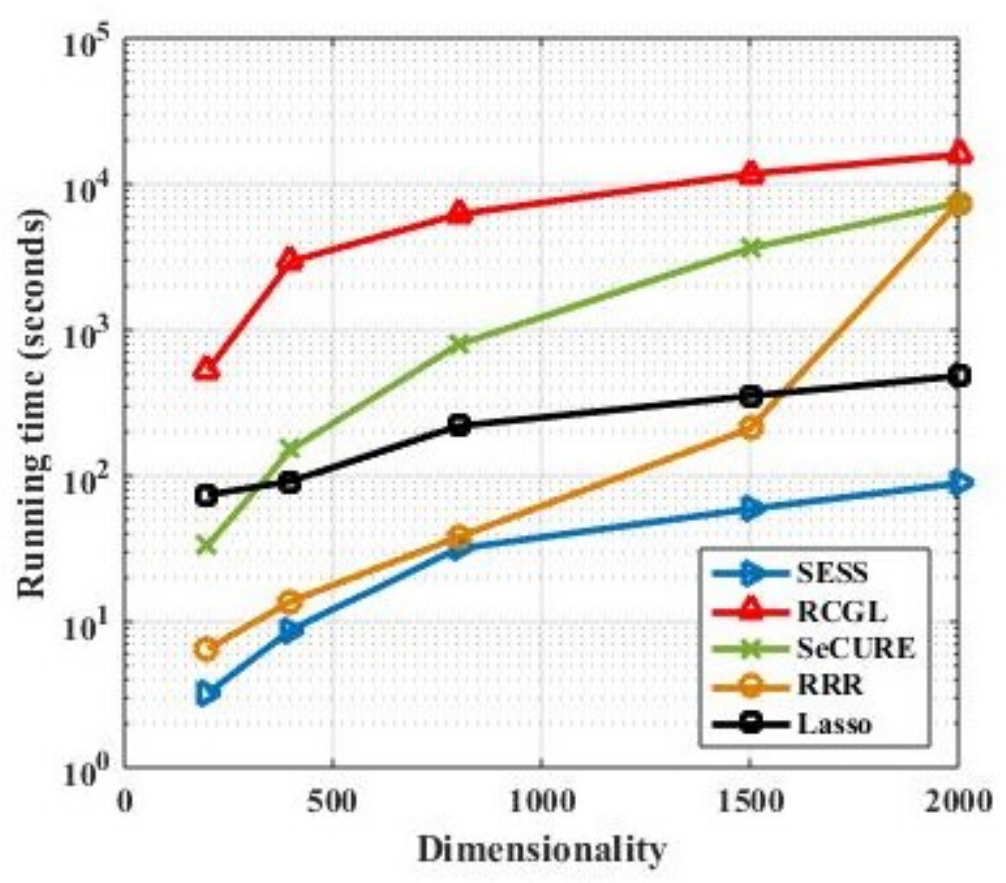}
  }
  \caption{CPU times of different methods}
  \label{fig:subfig} 
\end{figure}

\begin{figure} 
  \centering
  \subfigure[r = 3]{
    \label{label-a} 
    \includegraphics[width=2.5in]{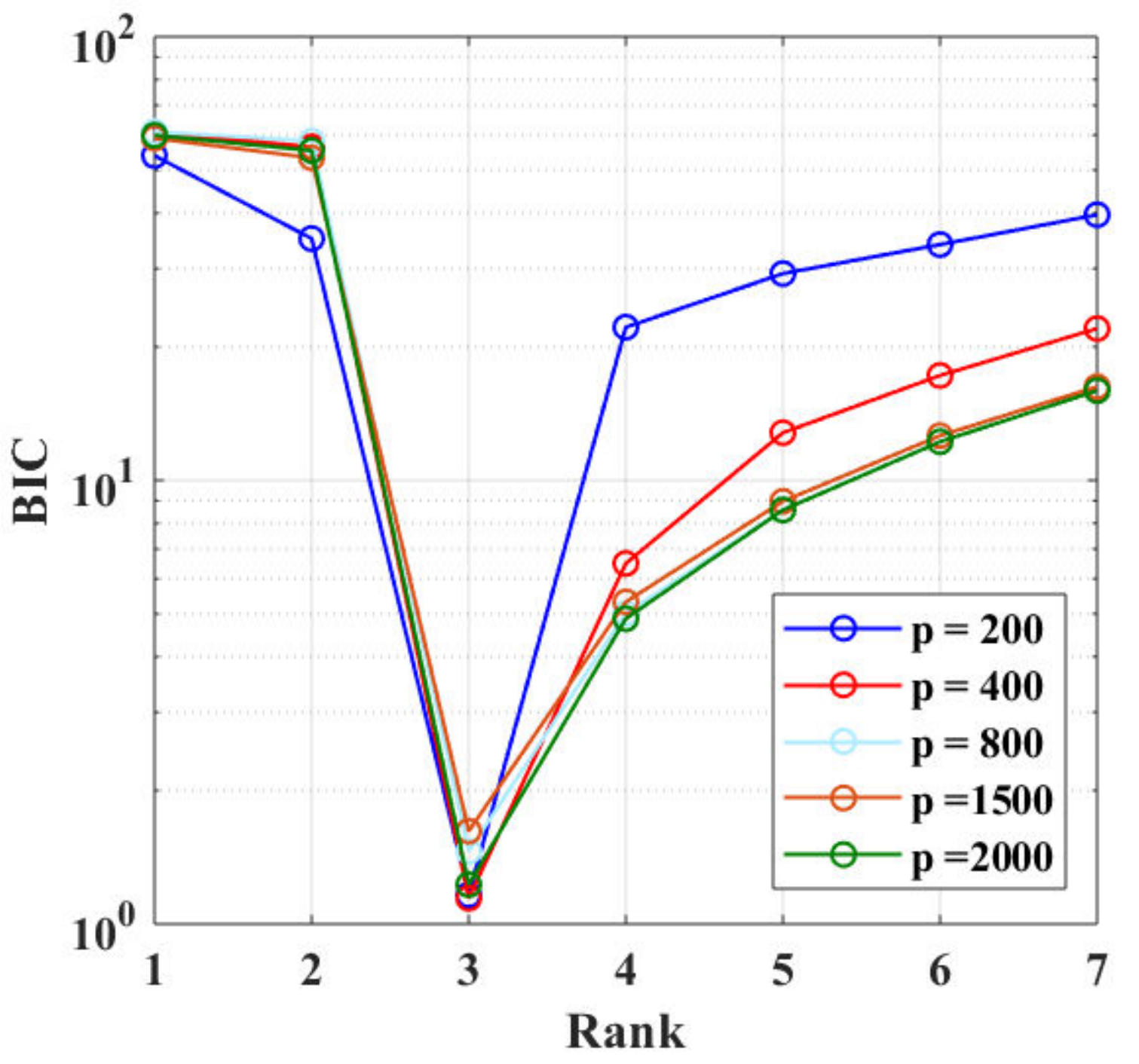}
  }
  \subfigure[r = 10]{
    \label{label-b} 
    \includegraphics[width=2.5in]{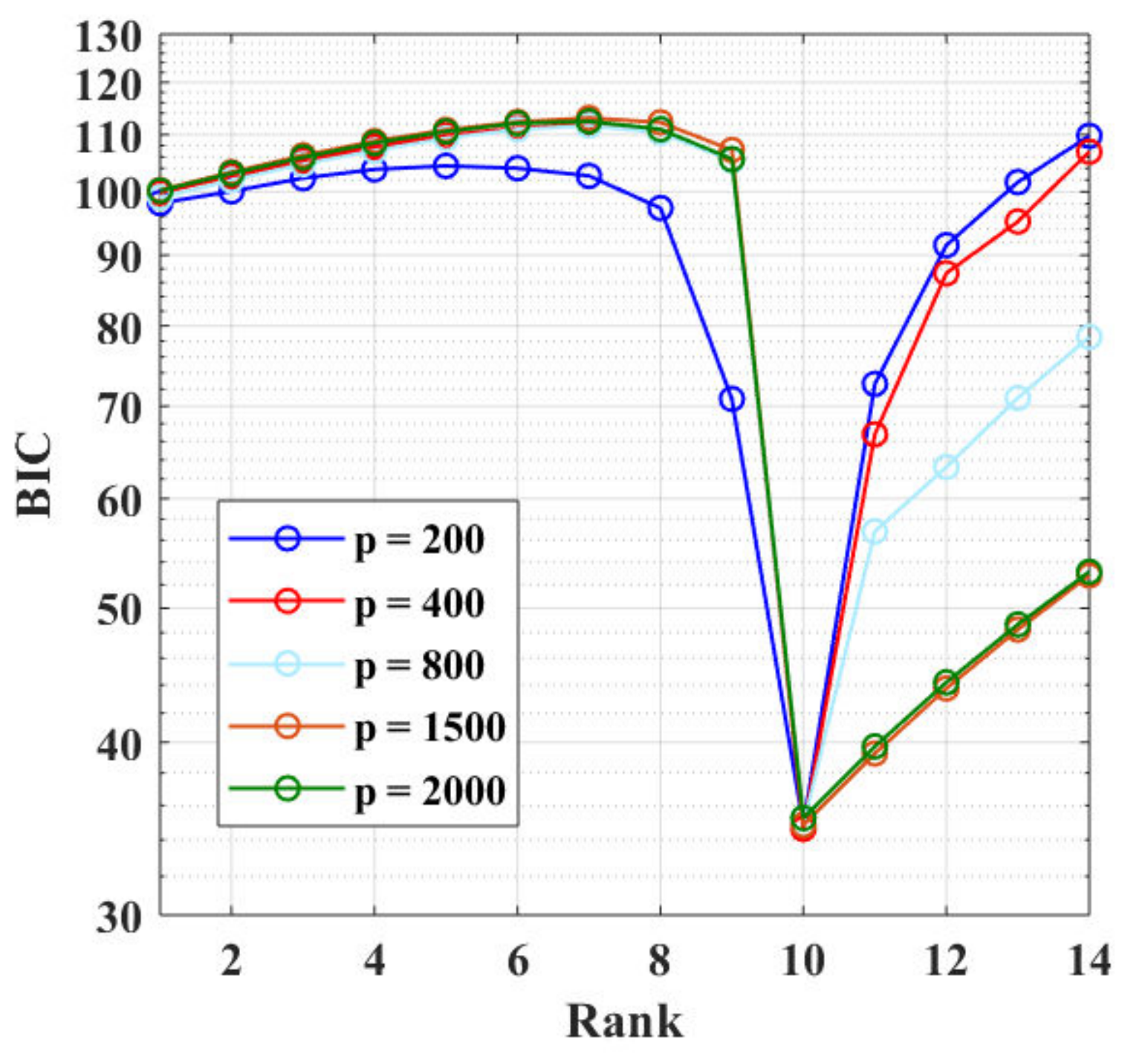}
  }
  \caption{BIC values in different steps for SESS}
  \label{p1} 
\end{figure}


Table \ref{simu1} summarizes the results of the performance measures except the CPU time. It is clear that the performance of SESS is among the best in terms of either prediction and estimation accuracies or variable selection under various settings. Although the computational efficiency of Lasso and RRR is good in view of Figure \ref{fig:subfig}, the Lasso can not recover and utilize the low rank structure, which in turn lowers its estimation and prediction accuracies, while the RRR suffers from the curse of dimensionality regardless of the correct identification of the rank. By contrast, SESS, RCGL, and SeCURE perform much better since they take advantage of the jointly low-rank and sparse structure.

Nevertheless, it can be seen from Figure \ref{fig:subfig} that SESS enjoys tremendous computational advantages over the other two comparable methods by increasing the speed for tens to hundreds of times, benefiting from its sequential formulation and tuning free property for the sparsity parameter. Specifically, when the dimensionality is 2000 and the true rank equals to 10, both RCGL and SeCURE need more than 2 hours to obtain the estimated coefficient matrix, while SESS costs less than 2 minutes in the same device. Furthermore, in view of Figure \ref{p1}, there are significant gaps between the BIC values of the true rank and other candidate ranks for SESS no matter how the dimensionality varies, which is due to the discrete nature of the rank. It makes SESS fairly easy to identify the correct one as it does not need to tune the sparsity level at the same time.

\subsection{Simulation example 2}\label{s2} 

In this second example, we generated 100 data sets and adopted the model setup similar to that in \cite{mishra2017sequential} with $(n, p, q, r) = (400, 500, 200, 3)$. Specifically, the true regression coefficient matrix $\bC^* = \sum_{j = 1}^{r^*} d_j^* \bu^*_j \bv_j^{*T}$ with the rank $r^* = 3$, $d_1^* = 60$, $d^*_2 = 30$, and $d^*_3 = 10$. The singular vectors were created as follows. We first generated $\breve{\bu}_1 = [ \rm{unif}(\mathcal{A}_u, s_1), \rm{rep}(0, \emph{p} - s_1)]^T$, $\breve{\bu}_2 = [ \rm rep(0, 5), \rm unif( \mathcal{A}_u, s_2), \rm rep(0, \emph{p} - 5 - s_2 ) ]^T$, and $\breve{\bu}_3 = [ \rm rep(0, 11), \rm unif( \mathcal{A}_u, s_3), \rm rep(0, \emph{p} - 11 - s_3 ) ]^T$, where $\mathcal{A}_u = \{1, -1\}$, $\rm unif(\mathcal{A}, \emph{b})$ denotes a $b$-dimensional vector whose entries are i.i.d. uniformly distributed on set $\mathcal{A}$, and $\rm rep(\alpha, \emph{k})$ denotes a $k$-dimensional vector whose entries are all equal to $\alpha$. Then we normalized them to have unit length, so that $\bu^*_j = \breve{\bu}_j /\|\breve{\bu}_j\|_2$, $j = 1, 2, 3$. Note that $s_k$ indicates the number of nonzero components in $\bu^*_k$ and can be various as displayed in Table \ref{simu3}. Similarly, we got $\breve{\bv}_1 = [ \rm unif( \mathcal{A}_v, 5), \rm rep(0, \emph{q} - 5) ]^T$, $\breve{\bv}_2 = [ \rm rep(0, 5), \rm unif( \mathcal{A}_v, 5), \rm rep(0, \emph{q} - 10) ]^T$, $\breve{\bv}_3 = [ \rm rep(0, 10), \rm unif( \mathcal{A}_v, 5), \rm rep(0, \emph{q} - 15) ]^T$ with $\rm\mathcal{A}_v = [-1, -0.3] \cup [0.3, 1]$, and $\bv^*_j = \breve{\bv}_j /\|\breve{\bv}_j\|_2$ for $j = 1, 2, 3$.

\begin{table}
	\begin{center} 
		\caption{Means and standard errors (in parentheses) of different performance measures in Section \ref{s2} } 
		
		\resizebox{450pt}{240pt}{
			\begin{tabular}{clccccc}
				Sparsity        &  Method        &  PE ($\times 10^{-2}$)  & EE ($\times 10^{-2}$)  & RE       & FNR       & FPR          \\    
\hline
				&Lasso                   & 21.68 (0.07)       & 10.08 (0.19)      & 95.40 (0.48)   & ------    & ------       \\
				$s_1 = 8$      &RRR                       & 39.56 (0.11)      & 42.88 (0.11)      & 0 (0)   & ------      &  ------      \\
				$s_2 = 9$       &SESS                      & 19.31 (0.04)       & 2.17 (0.01)      & 0 (0)   & 0 (0)      & 0.32 (0.01)       \\
				$s_3 = 9$       &RCGL                      & 19.39 (0.04)       & 2.35 (0.01)      & 0 (0)   & ------      & ------          \\
				&SeCURE                    & 19.44 (0.03)       & 2.91 (0.03)      & 0 (0)   & 0 (0)      & 0.36 (0.01)       \\
	
				\hline
				&Lasso                   & 22.12 (0.07)       & 10.98 (0.18)      & 97.00 (0.35)   & ------    & ------       \\
				$s_1 = 16$    &RRR                       & 40.44 (0.11)      & 44.05 (0.10)     & 0 (0)   & ------      &  ------      \\
				$s_2 = 18$    &SESS                      & 19.34 (0.04)       & 2.36 (0.01)      & 0 (0)   & 0 (0)      & 0.15 (0.01)       \\
				$s_3 = 18$    &RCGL                      & 19.55 (0.04)       & 2.79 (0.01)      & 0 (0)   & ------      & ------          \\
				&SeCURE                    & 19.38 (0.04)       & 2.38 (0.04)      & 0 (0)   & 0 (0)      & 0.16 (0.01)       \\
				\hline
				&Lasso                   & 22.16 (0.06)       & 12.91 (0.16)      & 98.40 (0.51)   & ------    & ------       \\
				$s_1 = 32$    &RRR                       & 41.52 (0.10)      & 47.00 (0.09)      & 0 (0)   & ------      &  ------      \\
				$s_2 = 36$    &SESS                      & 19.40 (0.05)       & 2.97 (0.01)      & 0 (0)   & 0 (0)      & 0.26 (0.01)       \\
		     	$s_3 = 36$    &RCGL                      & 19.47 (0.04)       & 3.33 (0.01)      & 0 (0)   & ------      & ------          \\
				&SeCURE                    & 19.48 (0.03)       & 3.31 (0.03)      & 0 (0)   & 0 (0)      & 0.31 (0.01)       \\
				\hline
				&Lasso                   & 22.67 (0.09)       & 16.60 (0.23)      & 93.82 (0.34)   & ------    & ------       \\
				$s_1 = 64$     &RRR                       & 41.73 (0.09)      & 43.19 (0.10)     & 0 (0)   & ------      &  ------      \\
				$s_2 = 72$     &SESS                      & 19.45 (0.05)       & 4.83 (0.04)      & 0 (0)   & 0 (0)      & 0.16 (0.01)       \\
				$s_3 = 72$     &RCGL                      & 19.99 (0.04)       & 6.68 (0.03)       & 0 (0)   & ------      & ------          \\
				&SeCURE                    & 19.14 (0.04)       & 4.90 (0.03)      & 0 (0)   & 0 (0)      & 0.16 (0.01)       \\
				\hline
				&Lasso                   & 30.12 (0.16)     & 31.91 (0.29)      & 92.40 (0.31)   & ------    & ------       \\
				$s_1 = 128$     &RRR                       & 40.13 (0.10)      & 44.34 (0.11)      & 0 (0)   & ------      &  ------      \\
				$s_2 = 144$     &SESS                      & 22.10 (0.09)      & 7.03 (0.10)      & 0 (0)   & 0.16 (0.09)      & 0.92 (0.06)       \\
				$s_3 = 144$     &RCGL                      & 23.44 (0.09)     & 8.25 (0.11)      & 0 (0)   & ------      & ------          \\
				&SeCURE                    & 23.92 (0.11)      & 8.56 (0.11)      & 0 (0)   & 0.15 (0.09)      & 0.95 (0.06)       \\
				\hline
				
			\end{tabular}
			
		}
		\label{simu3}
	\end{center}
	
\end{table}

Let $\bx$ follow the multivariate Gaussian distribution $N(\bzero, \mathbf{\Sigma}_X)$ with $\mathbf{\Sigma}_X = (0.5^{|i-j|})_{p \times p}$ and $\bx_1 = \bU^T \bx$ for some $\bU$ so that $\bx_1 \sim N(\bzero, \bI_{r^*})$. To generate the predictor matrix $\bX$, we first created $\bX_1 \in \mathbb{R}^{n \times r^*}$ by drawing $n$ random samples from $\bx_1 \sim N(\bzero, \bI_{r^*})$. Then based on $\bU^* = (\bu^*_1,\bu^*_2,\bu^*_3)$, we can find a $\bU^*_{\perp} \in \mathbb{R}^{p \times (p - r^*)}$ such that $\bP = (\bU^*, \bU^*_{\perp}) \in \mathbb{R}^{p \times p}$ and rank($\bP) = p$. Let $\bx_2 = \bU^{*T}_{\perp}\bx$ and $\bX_2 \in \mathbb{R}^{n \times (p-r^*)}$ was generated by drawing $n$ random samples from the conditional distribution of $\bx_2$ given $\bx_1$. Finally, the predictor matrix $\bX = (\bX_1, \bX_2)\bP^{-1}$. Moreover, we generated a non-Gaussian error matrix $\bE$ by first creating matrix $\breve{\bE}$ whose components are i.i.d. from a scaled $t$-distribution with $5$ degrees of freedom and unit variance, and then let $\bE = \sigma \breve{\bE} \mathbf{\Sigma}_E^{1/2}$ with $\mathbf{\Sigma}_E = (0.5^{|i-j|})_{q \times q}$. The noise level $\sigma$ is set so that the signal-to-noise ratio (SNR) defined as $\rm SNR$ $= \| d_{r^*}^* \bX \bu_{r^*}^* \bv_{r^*}^{*T} \|_F/$ $\| \bE \|_F$ equals to $0.75$. 



The results for different methods are summarized in Table \ref{simu3}. Similar to the first example, the performance of the jointly low-rank and sparse estimates including SESS, RCGL, and SeCURE is better than that of Lasso and RRR in terms of prediction and estimation accuracies and their performance is relatively stable regardless of the increasing number of nonzero components in $\bu^*_k$. Among them, SESS enjoys the highest computational efficiency similarly as in Section \ref{s1}. It also demonstrates the effectiveness of jointly low-rank and sparse estimation under some non-Gaussian errors.



\section{Application to stock short interest data} \label{real data} 
In this section, we will analyze the monthly stock short interest data set originally studied in \cite{rapach2016short}, available at Compustat (\url{http://www.hec.unil.ch/agoyal/}).
The raw data set reported short interest at the firm-level as the number of shares that were held short in a given firm. As short interest was shown in \cite{rapach2016short} to be the strongest predictor of aggregate stock returns, we will analyze the short interest influence networks among the firms and find the most influential firms through following vector auto-regression model with the maximal time lag $L$,
\begin{align*}
\by(t) = \sum_{i = 1}^{L} \bC_i^T \by(t-i) + \bepsilon(t).
\end{align*}
Here $\by(t) = \left[y_1(t), \dots, y_q(t) \right]^T \in \mathbb{R}^{q}$ consists of the short interests of $q$ firms at time $t$, $\bC_i \in \mathbb{R}^{q \times q}$ are the regression coefficient matrices, and $\bepsilon(t) \in \mathbb{R}^{q}$ denotes the random noise vector. By setting $\bx(t) = (\by(t-L)^T, \dots ,\by(t-1)^T)^T \in \mathbb{R}^{Lq}$ and $\bC = (\bC_1^T, \dots, \bC_L^T)^T$, the model can be rewritten as a multi-response regression model
\begin{align*}
\by(t) = \bC^T \bx(t) + \bepsilon(t).
\end{align*}


After the pre-processing, the data set consists of short interests of 3269 firms at the month level from January 1973 to December 2013, including 492 months in total. We set the maximal time lag $L = 5$ and the results are similar for larger lags. It yields a triple of $(n, p, q) = (487, 16345, 3269)$. Since both RCGL and SeCURE are no longer applicable due to the memory constraint in such large-scale data analysis, we report the performance of other methods in Section \ref{sec4}. By treating the first 366 samples as training data, we fit the multi-response regression model and then calculated the averaged $R^2$ statistics over $q$ firms for the significant latent factors and the averaged forecast error $(nq)^{-1}\| \bY - \bX \widehat{\bC } \|_F^2$ based on the remaining 121 testing samples. In view of the results summarized in Table \ref{tab:time}, SESS enjoys the lowest prediction error and identifies one significant latent factor. Its out-of-sample averaged $R^2$ statistic is as high as $15.20\%$, demonstrating its importance in forecasting the short interests. Moreover, there are 220 non-zero entries in the significant left singular vector with five entries much larger than others (at least 5 times larger). Correspondingly, the five most influential firms are two investment trust companies (Washington Prime Group Inc and Invesco) and three resource mining companies (Asanko Gold Inc, O'OKiep Copper, and Mesa Royalty Trust).

\begin{table}[!htb]
\caption{Results of different methods in Section \ref{real data}} 
\begin{center}
\begin{tabular}{cccccc}
Method              &  SESS        & Lasso   &RRR          \\ 
\hline
Estimated rank      &1             & 1019       &5            \\ 
Forecast error ($\times 10^{-3}$)    & 0.901        & 3.235     & 8.532       \\
$R^2$               &0.152         & 0.031       &0.008             \\
Time (minutes)      &4.950         & 45.623       & 64.125         \\
\end{tabular}
\label{tab:time}
\end{center}
\end{table}

Several researches reveal that market frictions and behavioral biases may cause price to deviate from fundamental value \citep{Miller 1977, Hong 1999} and that short sellers can exploit these situations since they are skilled at processing firm-specific information and information about future aggregate cash flows that is not reflected in current market prices \citep{Diether2009,rapach2016short}. Therefore, by applying our method to analyze the short interest of each company and their influence networks and forecast the behavior of short sellers, investors can make better judgments in response to market frictions to reasonably avoid certain risks.

\section{Discussion} \label{sec6}
In this paper, we have developed a new method SESS for high-dimensional multi-response regression, which recovers regression coefficient matrix and latent factors sequentially by converting the original problem into several univariate response regressions. Numerical studies demonstrate the statistical accuracy and high scalability of the proposed method. Our two-step sequential estimation procedure may be extended to deal with data containing measurement errors and outliers or more general model settings such as the generalized linear model, which will be interesting topics for future research.



\bibliographystyle{Chicago}

\newpage

\setcounter{page}{1}
\setcounter{section}{0}
\setcounter{equation}{0}

\renewcommand{\theequation}{A.\arabic{equation}}
\setcounter{equation}{0}

\bigskip
\begin{center}
{\large\bf Supplementary Material to ``Sequential scaled sparse factor regression''}

\bigskip

Zemin Zheng, Yang Li, Jie Wu and Yuchen Wang

\end{center}

%
%
%
%
\noindent This Supplementary Material presents the proofs for the theoretical results. \label{sec:app} 


\subsection*{Proof of Proposition ~\ref{ZYYZ}}
\smallskip
Note that $\bX\bu_k^*$ and $\widehat{\bZ}_k$ are the $k$th eigenvectors of the following two eigenvalue problems, respectively,
\begin{align}
\frac{1}{nq} \bY^* \bY^{*T} \bZ &= \lambda \bZ, \label{eig1}\\
\frac{1}{nq} \bY\bY^T \widehat{\bZ} &= \widehat{\lambda} \widehat{\bZ}. \label{eig2}
\end{align}
In order to show the uniform estimation error bound of $\widehat{\bZ}_k$, we will derive some bounds through random matrix theory (Lemma \ref{E}) to control the perturbation in the eigenvectors.

\smallskip

\textit{Step 1. Deriving the uniform bound on $\mid \widehat{\lambda}_k - \lambda_k \mid$}. Denote by $\Delta (\bY\bY)^T = \bY\bY^T - \bY^* \bY^{*T}$.
It follows directly from the eigenvalue perturbation theory that
\begin{align}\label{eig3}
\mid \widehat{\lambda}_k - \lambda_k \mid \leq \frac{1}{nq} \|\Delta (\bY\bY^T)\|_{2}.
\end{align}
Then we continue to analyze the term $\|\Delta (\bY\bY^T)\|_{2}$. By definition, we have
\[\Delta (\bY\bY^T) = \bY\bY^T - \bY^* \bY^{*T} = (\bY^* + \bE) (\bY^* + \bE)^T - \bY^* \bY^{*T} = \bY^* \bE^T + \bE\bY^{*T} + \bE\bE^T.\]
Applying the triangular inequality gives
\begin{align*}
\|\Delta (\bY\bY^T)\|_{2} & \leq \|\bY^* \bE^T\|_2 + \|\bE\bY^{*T}\|_2 + \|\bE\bE^T\|_2 \\ &\leq \|\bY^*\|_2 \| \bE^T \|_2 + \|\bE\|_2 \|\bY^{*T}\|_2 + \|\bE\|_2 \| \bE^T \|_2 \\
                               &= 2\|\bY^*\|_2 \| \bE \|_2 + \|\bE\|_2^2 \leq 2 \sqrt{nq\lambda_1} \gamma_u (2 \sqrt{n} + \sqrt{q}) + \gamma_u^2 (2 \sqrt{n} + \sqrt{q})^2, 
\end{align*}
where the last inequality holds with probability at least $1 - \exp(-\frac{n}{2})$ by setting $t = \sqrt{n}$ in Lemma \ref{E} such that $\|\bE\|_2 \leq \gamma_u (2\sqrt{n} +\sqrt{q})$.

\smallskip

Hereafter our discussion will be based on the event such that the upper bound on $\|\bE\|_2$ in Lemma \ref{E} holds and its probability is at least $1 - \exp(-\frac{n}{2})$. In view of inequality (\ref{eig3}), we get
\begin{align}\nonumber
\mid \widehat{\lambda}_k - \lambda_k \mid & \leq \frac{1}{nq} \|\Delta (\bY\bY^T)\|_{2} \leq 2 \sqrt{\lambda_1} \gamma_u (\frac{2 \sqrt{n} + \sqrt{q}}{\sqrt{nq}}) + \gamma_u^2 (\frac{2 \sqrt{n} + \sqrt{q}}{\sqrt{nq}})^2 \\\label{target}
                                          & < 4 \sqrt{\lambda_1} \gamma_u \frac{\sqrt{n} + \sqrt{q}}{\sqrt{nq}} + o(\frac{\sqrt{n} + \sqrt{q}}{\sqrt{nq}}).
\end{align}

\textit{Step 2. Deriving the uniform bound on $\|\widehat{\bZ}_k - \bX \bu_k^* \|_2/\sqrt{n}$}. Since the following argument applies to $\widehat{\bZ}_k$ and $\bX \bu_k^*$ with any fixed $k$, $1 \leq k \leq r^*$, we drop the index $k$ for notational clarity. Recall that $\|\widehat{\bZ}\|_2 = \|\bX \bu^* \|_2 = \sqrt{n}$. Then based on inequality (\ref{target}), applying the same argument as that in the proof of Lemma 6 \citep{bahadori2016scalable}, we can get
\begin{align*}
\frac{1}{\sqrt{n}} \|\widehat{\bZ} - \bX\bu^*\|_2 & < \frac{4 \sqrt{\lambda_1} \gamma_u} {d_\lambda} \cdot \frac{\sqrt{n} + \sqrt{q}}{\sqrt{nq}} + o(\frac{\sqrt{n} + \sqrt{q}}{\sqrt{nq}}).
\end{align*}
It is clear that the constant $\frac{4 \sqrt{\lambda_1} \gamma_u} {d_\lambda}$ is independent of index $k$. Then for sufficiently large $n$, we get
\begin{align}\label{ZZ}
\frac{1}{\sqrt{n}} \|\widehat{\bZ} - \bX\bu^*\|_2 & < \frac{4 \sqrt{\lambda_1} \gamma_u} {d_\lambda} \cdot \frac{\sqrt{n} + \sqrt{q}}{\sqrt{nq}}.
\end{align}
Thus, the above inequality provides a uniform estimation error bound on $\|\widehat{\bZ}_k - \bX \bu_k^* \|_2/\sqrt{n}$. It concludes the proof of Proposition \ref{ZYYZ}.

\smallskip

\subsection*{Proof of Corollary ~\ref{nds}}

\smallskip

The key point of this proof is to quantify the upper bound of $\|\bE\|_2$ under sub-Gaussian distribution. By setting $t = \sqrt{n+2\log 2}$ in Lemma \ref{Es}, it yields that with probability at least $1 - \exp(-\frac{n}{2})$,
\begin{align*}
\|\bE\|_2 \leq \gamma^{\star} (\sqrt{n} +\sqrt{q}+\sqrt{n+2\log 2}) \leq\gamma^{\star} (\sqrt{n} +\sqrt{q} +(1+\log2)\sqrt{n})\leq\gamma_u' (2\sqrt{n} +\sqrt{q}),
\end{align*}
where $\gamma_u'= (1+\log2)\gamma^{\star}\geq2(1+\log2)$ is some constant. Then applying the same argument as that in the last section (the proof of Proposition ~\ref{ZYYZ}) gives the results of Corollary ~\ref{nds}.

\smallskip

\subsection*{Proof of Theorem~\ref{theo1}}

\smallskip


The following proof is conditional on the event such that the results in Theorem \ref{ZYYZ} hold. Similar to the proof of Theorem \ref{ZYYZ}, since the following argument applies to any singular vector and unit rank matrix for $1 \leq k \leq r^*$, we drop the index $k$ for notational clarity. The bounds on the five quantities in Theorem~\ref{theo1} will be derived in three steps.
\smallskip

\textit{Step 1. Deriving the uniform bounds on $\|\widehat{\bu} - \bu^*\|_2$ and $\|\bX\widehat{\bu} - \bX \bu^*\|_2/\sqrt{n}$}. For $1 \leq j \leq p$, denote by $\bX_{j}$ the $j$th column of $\bX$. Since the columns of $\bX$ are standardized to have a common $L_2$-norm $\sqrt{n}$, we get $\|\bX_{j}\|_{2} = \sqrt{n}$. Therefore, conditional on the event such that the results of Theorem \ref{ZYYZ} hold, we have 
\begin{align*} 
   \frac{1}{n} \|\bX^{T}(\widehat{\bZ} - \bX\bu^*)\|_{\infty} & = \max\limits_{1 \leq j\leq p} \frac{1}{n} \mid \bX^T_{j}(\widehat{\bZ} - \bX\bu^*)\mid \nonumber 
   \leq \max\limits_{1 \leq j\leq p} \frac{\|\bX_{j}\|_{2}}{\sqrt{n}} \cdot \frac{\|\widehat{\bZ} - \bX\bu^*\|_{2}}{\sqrt{n}} \\
   & = \frac{\|\widehat{\bZ} - \bX\bu^*\|_{2}}{\sqrt{n}} < \frac{4 \sqrt{\lambda_1} \gamma_u} {d_\lambda} \cdot \frac{\sqrt{n} + \sqrt{q}}{\sqrt{nq}} + o(\frac{\sqrt{n} + \sqrt{q}}{\sqrt{nq}}). 
\end{align*}
When $\omega = \widetilde{C} (\frac{\sqrt{n} + \sqrt{q}}{\sqrt{nq}}) (\frac{\xi +1}{\xi -1})$ with $\widetilde{C} > \frac{4 \sqrt{\lambda_1} \gamma_u} {d_\lambda}$, for sufficiently large $n$, it yields that 
\begin{align*}
 \frac{\|\bX^{T}(\widehat{\bZ} - \bX\bu^*)\|_{\infty}}{n} \leq \omega \frac{(\xi - 1)}{(\xi+1)}.
\end{align*}

Thus, by the same argument as the proof of Theorem 3 \citep{ye2010rate}, we can get 
\begin{align}\label{hi}
    & \|\widehat{\bu} - \bu^* \|_{1}  \leq  \frac{2 \xi s \omega}{(\xi + 1)F_1(\xi,S)}, \nonumber \\
    \|\widehat{\bu} - \bu^* \|_{2} & \leq  \frac{2 \xi s^{1/2} \omega}{(\xi + 1)F_2(\xi,S)} \leq C_{u}\sqrt{s} \Big(\frac{\sqrt{n} + \sqrt{q}}{\sqrt{nq}}\Big),
\end{align}
where $C_{u} = \frac{2 \widetilde{C} \xi }{(\xi - 1) F_2}$ and $F_2(\xi,S) \geq F_2$ under Condition \ref{cond1}. Further applying the triangular inequality and inequality (23) in \cite{sun2012scaled}, which is derived through the Karush$-$Kuhn$-$Tucker condition, gives
\begin{align*}
 & \frac{2}{n} \|\bX\widehat{\bu} - \bX\bu^*\|^2_2 \leq 2 \omega (\|\bu^*\|_1 - \|\widehat{\bu}\|_1 ) + \frac{2}{n}\|\bX^T (\widehat{\bZ} - \bX \bu^*)\|_\infty \cdot \| \bu^* - \widehat{\bu}\|_1 \\
                              & \leq 2 \omega \| \bu^* - \widehat{\bu}\|_1 \left( 1 + \frac{\xi - 1}{\xi+1} \right) = \frac{4 \omega \xi}{\xi+1} \| \widehat{\bu} - \bu^*\|_1 \leq \frac{8 \omega^2 \xi^2 s}{(\xi + 1)^2 F_1(\xi,S)}. 
\end{align*}

Under Condition \ref{cond1}, it follows that
\begin{align}\label{XU}
\frac{1}{\sqrt{n}} \|\bX\widehat{\bu} - \bX \bu^*\|_2 \leq \frac{2 \xi s^{1/2} \omega}{(\xi + 1) \sqrt{F_1(\xi,S)}} \leq \widetilde{C}_{u}\sqrt{s} \Big(\frac{\sqrt{n} + \sqrt{q}}{\sqrt{nq}}\Big),
\end{align}
where $\widetilde{C}_{u} = \frac{2 \widetilde{C} \xi }{(\xi - 1) \sqrt{F_1}}$.

\smallskip

\textit{Step 2. Deriving the uniform bound on $\|\widehat{\bv} - \bv^*\|_2/\sqrt{q}$}.
According to (\ref{eq:solu}), with the estimated latent factor $\widehat{\bZ}$, the corresponding right singular vector $\widehat{\bv}$ is estimated as
\begin{equation*}
\widehat{\bv} = \frac{1}{n}\bY^T \widehat{\bZ},
\end{equation*}
which is motivated by the intrinsic relationship between $\bZ^*$ and $\bv^*$ in the noiseless case that
\begin{equation*}
\bv^* = \frac{1}{n}\bY^{*T} \bZ^*.
\end{equation*}

Since $\bY = \bY^* + \bE$, we can analyze the difference between them as
\begin{align*}
\|\widehat{\bv} - \bv^{*}\|_2 & = \frac{1}{n} \|(\bY^* + \bE)^T \widehat{\bZ} - \bY^{*T} \bZ^* \|_2 \\
                    & \leq \frac{\|\bY^*\|_2 \cdot \|\widehat{\bZ} - \bZ^*\|_2}{n} + \frac{\|\bE\|_2 \cdot \|\widehat{\bZ}\|_2}{n}.
\end{align*}
It follows from $\|\bY^*\|_2 = \sqrt{nq \lambda_1}$, $\|\widehat{\bZ}\|_2 = \sqrt{n}$, $\|\bE\|_2 \leq \gamma_u (2\sqrt{n} +\sqrt{q})$, and inequality (\ref{ZZ}) that
\begin{align*}
\frac{1}{\sqrt{q}} \|\widehat{\bv} - \bv^{*}\|_2 & < \frac{4\lambda_1 \gamma_u}{d_{\lambda}} \Big(\frac{\sqrt{n} + \sqrt{q}}{\sqrt{nq}}\Big) + 2\gamma_u \Big(\frac{\sqrt{n} + \sqrt{q}}{\sqrt{nq}}\Big) + o\left(\frac{\sqrt{n} + \sqrt{q}}{\sqrt{nq}} \right)\\
& < ( \frac{4\lambda_1 \gamma_u}{d_{\lambda}} + 2\gamma_u) \Big(\frac{\sqrt{n} + \sqrt{q}}{\sqrt{nq}}\Big) + o\left(\frac{\sqrt{n} + \sqrt{q}}{\sqrt{nq}} \right).
\end{align*}

Thus, for sufficiently large $n$, we have
\begin{align}\label{e2} 
\frac{1}{\sqrt{q}} \|\widehat{\bv} - \bv^{*}\|_2  < C_{v} \Big(\frac{\sqrt{n} + \sqrt{q}}{\sqrt{nq}}\Big),
\end{align}
where $C_{v} = (4\lambda_1/d_{\lambda} + 2)\gamma_u$. 

\smallskip

\textit{Step 3. Deriving the uniform bounds on $\|\widehat{\bC} - \bC^*\|_2/\sqrt{q}$ and $\|\bX\widehat{\bC} - \bX\bC^*\|_2/\sqrt{nq}$}. By definitions of the unit rank matrices $\bC^*$ and $\widehat{\bC}$, we have
\[\bC^* - \widehat{\bC} = \bu^* \bv^{*T} - \widehat{\bu} \widehat{\bv}^T = (\bu^* - \widehat{\bu})\bv^{*T} + \widehat{\bu}(\bv^{*} - \widehat{\bv})^T.\]
Then it follows from Condition \ref{conduv} and the estimation error bounds (\ref{hi}) and (\ref{e2}) that
\begin{align*}
\frac{1}{\sqrt{q}}\|\bC^* - \widehat{\bC}\|_F &\leq \frac{1}{\sqrt{q}}\|(\bu^* - \widehat{\bu})\bv^{*T}\|_F + \frac{1}{\sqrt{q}}\|\widehat{\bu}(\bv^{*} - \widehat{\bv})^T\|_F \\
& = \frac{1}{\sqrt{q}} \|\bu^* - \widehat{\bu}\|_2 \cdot \|\bv^{*}\|_2 + \frac{1}{\sqrt{q}}\|\widehat{\bu}\|_2 \cdot \|\bv^{*} - \widehat{\bv}\|_2\\
& < V C_{u} \sqrt{s} \Big(\frac{\sqrt{n} + \sqrt{q}}{\sqrt{nq}}\Big) + U C_{v} \Big(\frac{\sqrt{n} + \sqrt{q}}{\sqrt{nq}}\Big) + o\left( \sqrt{s}\frac{\sqrt{n} + \sqrt{q}}{\sqrt{nq}} \right),
\end{align*}
where the last inequality utilizes the triangular inequality $\|\widehat{\bu}\|_2 \leq \|\bu^*\|_2 + \|\widehat{\bu} - \bu^*\|_2$. Thus, for sufficiently large $n$, we obtain
\begin{align*}
\frac{1}{\sqrt{q}}\|\widehat{\bC} - \bC^*\|_F < (V C_{u} + U C_{v}) \sqrt{s} \Big(\frac{\sqrt{n} + \sqrt{q}}{\sqrt{nq}}\Big).
\end{align*}

Similarly, for the prediction error bound of the unit rank matrix, with the inequality (\ref{XU}) and sufficiently large $n$, we have 
\begin{align*}
 \frac{1}{\sqrt{nq}} \|\bX(\bC^* - \widehat{\bC})\|_F &\leq \frac{1}{\sqrt{nq}}\|(\bX\bu^* - \bX\widehat{\bu})\bv^{*T}\|_F + \frac{1}{\sqrt{nq}}\|\bX\widehat{\bu}(\bv^{*} - \widehat{\bv})^T\|_F \nonumber\\
& = \frac{1}{\sqrt{nq}} \|\bX(\bu^* - \widehat{\bu})\|_2 \cdot \|\bv^{*}\|_2 + \frac{1}{\sqrt{nq}}\|\bX \widehat{\bu}\|_2 \cdot \|\bv^{*} - \widehat{\bv}\|_2 \nonumber\\
&< (V \widetilde{C}_{u} + C_{v}) \sqrt{s} \Big(\frac{\sqrt{n} + \sqrt{q}}{\sqrt{nq}}\Big). 
\end{align*}
By the same argument, we get
\begin{align}
 \frac{1}{\sqrt{nq}} \|\widehat{\bZ} \widehat{\bv}^T - \bX  \bC^*_k\|_F &\leq \frac{1}{\sqrt{nq}}\|(\widehat{\bZ} - \bZ^*)\bv^{*T}\|_F + \frac{1}{\sqrt{nq}}\|\widehat{\bZ}(\widehat{\bv} - \bv^{*})^T\|_F \nonumber\\
& = \frac{1}{\sqrt{nq}} \|\widehat{\bZ} - \bZ^*\|_2 \cdot \|\bv^{*}\|_2 + \frac{1}{\sqrt{nq}}\|\widehat{\bZ}\|_2 \cdot \|\widehat{\bv} - \bv^{*}\|_2 \nonumber\\
&< (V \widetilde{C} + C_{v}) \Big(\frac{\sqrt{n} + \sqrt{q}}{\sqrt{nq}}\Big). \label{Xu}
\end{align}
It completes the proof of Theorem \ref{theo1}.

\smallskip

\subsection*{Proof of Theorem~\ref{theo2}}

\smallskip

Before showing the results of Theorem \ref{theo2}, some preparations are needed. Since $\widehat{\bZ}_k$ is the eigenvector of equation (\ref{eig2}) with respect to the $k$th eigenvalue $\widehat{\lambda}_k$, we have $\frac{1}{nq} \bY\bY^T \widehat{\bZ}_k = \widehat{\lambda}_k \widehat{\bZ}_k$.
It follows that
\begin{align*}
\widehat{\lambda}_k = \frac{\widehat{\bZ}_k^T \bY\bY^T \widehat{\bZ}_k}{nq \widehat{\bZ}_k^T \widehat{\bZ}_k} = \frac{\widehat{\bZ}_k^T \bY\bY^T \widehat{\bZ}_k}{n^2 q}.
\end{align*}
Moreover, when tuning the rank, the corresponding right singular vector $\widehat{\bv}_k$ can be obtained by 
\begin{align*} 
\widehat{\bv}_k = \frac{1}{n} \bY^T \widehat{\bZ}_k.
\end{align*}

We first analyze the loss function $\mathcal{L}(k) = \frac{1}{nq}\|\bY - \widehat{\bY}_{k}\|_F^2$ and the information criterion $\mathcal{C}({k})$ in the $k$th step. Since $\widehat{\bY}_{k} = \sum_{j=1}^{k} \widehat{\bZ}_j \widehat{\bv}_j^T$, the amount of decrease in $\mathcal{L}(k)$ in the $k$th step satisfies that
\begin{align*}
\mathcal{L}(k - 1) - \mathcal{L}(k) & = \frac{1}{nq}\|\bY - \widehat{\bY}_{k - 1}\|_F^2 - \frac{1}{nq}\|\bY - \widehat{\bY}_{k}\|_F^2 \\ &= \frac{1}{nq}\|\bY - \widehat{\bY}_{k - 1}\|_F^2 - \frac{1}{nq}\|\bY - \widehat{\bY}_{k - 1} - \widehat{\bZ}_k \widehat{\bv}_k^{T}\|_F^2 \\
& = \frac{1}{nq}\left(2\left\langle \bY - \widehat{\bY}_{k - 1},  \widehat{\bZ}_k \widehat{\bv}_k^{T}\right\rangle - \|\widehat{\bZ}_k \widehat{\bv}_k^{T}\|_F^2 \right) =  \frac{1}{nq}\left(2\left\langle \bY , \widehat{\bZ}_k \widehat{\bv}_k^{T}\right\rangle - \|\widehat{\bZ}_k \widehat{\bv}_k^{T}\|_F^2 \right),
\end{align*}
where the last equality holds due to the orthogonality of different eigenvectors. 
Replacing $\widehat{\bv}_k$ with $\frac{1}{n} \bY^T \widehat{\bZ}_k$, we further have
\begin{align}\label{lam}
\mathcal{L}(k - 1) - \mathcal{L}(k) & = \frac{1}{nq} \left( \frac{2 \widehat{\bZ}_k^{T} \bY\bY^T \widehat{\bZ}_k}{n}
  - \frac{\widehat{\bZ}_k^{T} \bY\bY^T \widehat{\bZ}_k}{n}\right) = \frac{1}{n^2q} \widehat{\bZ}_k^{T} \bY\bY^T \widehat{\bZ}_k = \widehat{\lambda}_k.
\end{align}

Moreover, by the definition of $\mathcal{C}(k)$, we get $\mathcal{C}({k - 1}) - \mathcal{C}({k}) = \sqrt{n} \log (\mathcal{L}(k - 1)/ \mathcal{L}(k)) - \log n$. Both lower and upper bounds on $\log (\mathcal{L}(k - 1) / \mathcal{L}(k))$ can be provided 
by observing the fact that $1 - \frac{1}{x}\leq \log(x) \leq x-1 $ for $x >0$, so that
\begin{align}\label{log}
\frac{\mathcal{L}(k - 1) - \mathcal{L}(k)}{\mathcal{L}(k - 1)} \leq \log \left( \frac{\mathcal{L}(k - 1)}{\mathcal{L}(k)}\right) \leq \frac{\mathcal{L}(k - 1) - \mathcal{L}(k)}{\mathcal{L}(k)}.
\end{align}
Then the results of Theorem \ref{theo2} will be shown in two steps.

\smallskip

\textit{Step 1. We show that $\mathcal{C}({k - 1}) > \mathcal{C}({k})$ when $1 \leq k \leq r^*$}. According to the uniform estimation error bound of population eigenvalues in (\ref{target}), conditional on the event such that the results of Theorem \ref{ZYYZ} hold, for sufficiently large $n$, we have
\begin{align}\label{lamd}
\mid \widehat{\lambda}_k - \lambda_k \mid & < C \Big(\frac{\sqrt{n} + \sqrt{q}}{\sqrt{nq}}\Big) 
\end{align}
for the positive constant $C = 4 \sqrt{\lambda_1} \gamma_u$. Then for any $1 \leq k \leq r^*$, it follows from (\ref{lam}) that
\begin{align}\label{boundL}
\mathcal{L}(k - 1) - \mathcal{L}(k) = \widehat{\lambda}_k > \lambda_k - C \Big(\frac{\sqrt{n} + \sqrt{q}}{\sqrt{nq}}\Big).
\end{align}
Moreover, by the fact that $\bY = \sum_{j=1}^{r^*} \bX \bu_j^* \bv_j^{*T} + \bE$ and applying the triangular inequality, we have
\begin{align*}
\sqrt{\mathcal{L}(k - 1)} = \frac{1}{\sqrt{nq}}\|\bY - \widehat{\bY}_{k - 1}\|_F = \frac{1}{\sqrt{nq}} \big\|\sum_{j=1}^{r^*} \bX \bu_j^* \bv_j^{*T} + \bE - \sum_{j=1}^{k - 1} \widehat{\bZ}_j \widehat{\bv}_j^T \big\|_F \\
\leq \frac{1}{\sqrt{nq}}\Big(\sum_{j=1}^{k-1}\|\bX \bu_j^* \bv_j^{*T} - \widehat{\bZ}_j \widehat{\bv}_j^T\|_F + \sum_{j=k}^{r^*}\|\bX \bu_j^* \bv_j^{*T}\|_F + \|\bE\|_F \Big).
\end{align*}
We will bound the three terms on the right hand side successively.

For the first term, by the results of Theorem \ref{ZYYZ} and applying the same arguments as (\ref{e2}) and (\ref{Xu}) in the proof of Theorem \ref{theo1}, 
we can show that uniformly over $1 \leq j \leq k - 1$, for sufficiently large $n$,
\begin{align*} 
 \frac{1}{\sqrt{nq}}  & \|\bX \bu_j^* \bv_j^{*T} - \widehat{\bZ}_j \widehat{\bv}_j^T\|_F  < \widetilde{C}_z \Big(\frac{\sqrt{n} + \sqrt{q}}{\sqrt{nq}}\Big),
\end{align*}
where the positive constants $C_{v} = (4\lambda_1/d_{\lambda} + 2)\gamma_u$ and $\widetilde{C}_z = 4 \sqrt{\lambda_1} \gamma_{u} V /d_\lambda + C_v$. It gives that
\begin{align*}
\sum_{j=1}^{k-1} \frac{\|\bX \bu_j^* \bv_j^{*T} - \widehat{\bZ}_j \widehat{\bv}_j^T\|_F}{\sqrt{nq}} < (k - 1) \widetilde{C}_z \Big(\frac{\sqrt{n} + \sqrt{q}}{\sqrt{nq}}\Big). 
\end{align*}
For the second term, since $\|\bX \bu_j^* \bv_j^{*T}\|_F/\sqrt{nq} = \sqrt{\lambda}_j$, we have
\begin{align*}
\frac{1}{\sqrt{nq}}\sum_{j=k}^{r^*}\|\bX\bu_j^* \bv_j^{*T}\|_F = \sum_{j=k}^{r^*} \sqrt{\lambda}_j.
\end{align*}

Then we bound the last term. As the components of $\bE \mathbf{\Sigma}^{-1/2}$ are independent and identically distributed with the standard Gaussian distribution, given the tail bound for $\chi^2$ distribution \citep[Lemma 1]{laurent2000adaptive}, we have with probability at least $1 - \exp(-\frac{n}{2})$,
\begin{align*}
\|\bE \mathbf{\Sigma}^{-1/2}\|_F^2/nq  \leq 1 + \sqrt{\frac{2}{q}} + \frac{1}{q} < \left(1 + \frac{1}{\sqrt{q}} \ \right)^2.
\end{align*}
On the other hand, by Condition \ref{cond3}, we have
\begin{align*}
\|\bE \mathbf{\Sigma}^{-1/2}\|_F^2 \geq \|\bE\|_F^2 \lambda_{\min}^2(\mathbf{\Sigma}^{-1/2}) \geq \|\bE\|_F^2/\gamma_{u}^2.
\end{align*}
These two inequalities together yield
\begin{align*}
\|\bE\|_F/\sqrt{nq}  \leq \gamma_u \left(1 + \frac{1}{\sqrt{q}} \right).
\end{align*}

Combining the three bounds gives
\begin{align*}
\sqrt{\mathcal{L}(k - 1)} & < \sum_{j=k}^{r^*} \sqrt{\lambda}_j + \gamma_u \left(1 + \frac{1}{\sqrt{q}} \right)  +  (k - 1) \widetilde{C}_z \Big(\frac{\sqrt{n} + \sqrt{q}}{\sqrt{nq}}\Big) \\
& \leq (r^* - k + 1 + c_{\gamma}) \sqrt{\lambda}_k + (\gamma_u + \widetilde{C}_z r^*) \Big(\frac{\sqrt{n} + \sqrt{q}}{\sqrt{nq}}\Big),
\end{align*}
where the constant $c_{\gamma} = \gamma_u/ \sqrt{\lambda}_{r^*}$ so that $\gamma_u = c_{\gamma} \sqrt{\lambda}_{r^*} \leq c_{\gamma} \sqrt{\lambda}_k$. It holds with probability at least $1 - 2\exp(-\frac{n}{2})$ by applying the union bound to control the tail probability of the union of the two events. Our discussion will be conditioning on this new event hereafter. 

Together with inequality (\ref{boundL}), we can derive that
\begin{align*}
\sqrt{\frac{\mathcal{L}(k - 1) - \mathcal{L}(k)}{\mathcal{L}(k - 1)}} & > \frac{\sqrt{\lambda}_k + O\left(\frac{\sqrt{n} + \sqrt{q}}{\sqrt{nq}} \right)}{(r^* - k + 1 + c_{\gamma}) \sqrt{\lambda}_k + O\left(r^* \frac{\sqrt{n} + \sqrt{q}}{\sqrt{nq}} \right)}\\
& = \frac{1}{r^* - k + 1 + c_{\gamma}} + O\left(r^* \frac{\sqrt{n} + \sqrt{q}}{\sqrt{nq}} \ \right).
\end{align*}
Then it follows from the assumption ${r^*} (\frac{\log n}{\sqrt{n}})^{1/2} = o(1)$, which implies ${r^*} (\frac{\sqrt{n} + \sqrt{q}}{\sqrt{nq}})^{1/2} = o(1)$, that for sufficiently large $n$,
\begin{align*}
\sqrt{\frac{\mathcal{L}(k - 1) - \mathcal{L}(k)}{\mathcal{L}(k - 1)}} >  \frac{1}{r^* - k + 1 + c_{\gamma}}.
\end{align*}

In view of (\ref{log}), it leads to
\begin{align*}
& \mathcal{C}({k - 1}) - \mathcal{C}({k}) = \sqrt{n} \log (\mathcal{L}(k - 1)/ \mathcal{L}(k)) - \log n\\
& > \frac{\sqrt{n}}{(r^* - k + 1 + c_{\gamma})^2} - \log n \geq \frac{\sqrt{n}}{(r^*+ c_{\gamma})^2} - \log n > 0.
\end{align*}
It means that the information criterion $\mathcal{C}(k)$ will keep decreasing when the sequential step $k$ is no more than the true rank $r^*$.

\smallskip

\textit{Step 2. We show that $\mathcal{C}({k - 1}) < \mathcal{C}({k})$ when $k > r^*$}. Since $\lambda_k = 0$ when $k > r^*$, by the uniform estimation error bound (\ref{lamd}), we have
\begin{align}\label{sig}
\mathcal{L}(k - 1) - \mathcal{L}(k) = \widehat{\lambda}_k < C \Big(\frac{\sqrt{n} + \sqrt{q}}{\sqrt{nq}}\Big).
\end{align}
Moreover, by the triangular inequality we have
\begin{align*}
\sqrt{\mathcal{L}(k)} = \frac{1}{\sqrt{nq}}\|\bY - \widehat{\bY}_{k}\|_F \geq \frac{1}{\sqrt{nq}}\Big(\|\bE\|_F - \sum_{j=1}^{r^*}\|\bX \bu_j^* \bv_j^{*T} - \widehat{\bZ}_j \widehat{\bv}_j^T\|_F - \sum_{j = r^* + 1}^{k} \|\widehat{\bZ}_j \widehat{\bv}_j^T\|_F\Big).
\end{align*}

Applying similar arguments as in \textit{Step 1}, the three terms on the right hand side can be bounded as
\begin{align*} 
 & \frac{1}{\sqrt{nq}} \|\bE\|_F \geq \gamma_l \left(1 - \frac{2}{\sqrt{q}}  \right),\\ 
 & \frac{1}{\sqrt{nq}} \sum_{j=1}^{r^*}\|\bX \bu_j^* \bv_j^{*T} - \widehat{\bZ}_j \widehat{\bv}_j^T\|_F < \widetilde{C}_z r^* \left( \frac{\sqrt{n} + \sqrt{q}}{\sqrt{nq}} \right),\\
 & \frac{1}{\sqrt{nq}} \sum_{j = r^* + 1}^{k} \|\widehat{\bZ}_j \widehat{\bv}_j^T\|_F = \sum_{j = r^* + 1}^{k} \sqrt{\widehat{\lambda}}_j < (r - r^*) \sqrt{C} \left( \frac{\sqrt{n} + \sqrt{q}}{\sqrt{nq}} \right)^{1/2}.
\end{align*}
Combining the above results yields the following upper bound
\begin{align*}
\sqrt{\frac{\mathcal{L}(k - 1) - \mathcal{L}(k)}{\mathcal{L}(k)}} & < \frac{ \sqrt{C} \left( \frac{\sqrt{n} + \sqrt{q}}{\sqrt{nq}}\right)^{1/2}}{\gamma_l - (\widetilde{C}_z + 2\gamma_l) \left( r^* \frac{\sqrt{n} + \sqrt{q}}{\sqrt{nq}} \right) - (r - r^*) \sqrt{C} \left( \frac{\sqrt{n} + \sqrt{q}}{\sqrt{nq}} \right)^{1/2}}.
\end{align*}

Since $r^* \left( \frac{\sqrt{n} + \sqrt{q}}{\sqrt{nq}} \right)^{1/2} = o(1)$ and $r \left( \frac{\sqrt{n} + \sqrt{q}}{\sqrt{nq}} \right)^{1/2} = o(1)$, for sufficiently large $n$, we have
\begin{align}\label{kbr}
\sqrt{\frac{\mathcal{L}(k - 1) - \mathcal{L}(k)}{\mathcal{L}(k)}} & < \frac{\sqrt{C}}{\gamma_l} \left( \frac{\sqrt{n} + \sqrt{q}}{\sqrt{nq}} \right)^{1/2}.
\end{align}
In view of (\ref{log}), it leads to
\begin{align*}
\mathcal{C}({k - 1}) & - \mathcal{C}({k}) = \sqrt{n} \log (\mathcal{L}(k - 1)/ \mathcal{L}(k)) - \log n < \frac{C}{\gamma_l^2} \left( \frac{\sqrt{n} + \sqrt{q}}{\sqrt{q}} \right) - \log n < 0, 
\end{align*}
where the last inequality is immediate from the assumption that $\frac{\sqrt{n}}{\sqrt{q} \log n} = o (1)$.

Combining the established results in the aforementioned two steps gives that $\mathcal{C}({k})$ will attain its minimum value when $k = r^*$ with probability at least $ 1 - 2\exp(-\frac{n}{2})$ for sufficiently large $n$, which concludes the proof of Theorem \ref{theo2}.

\subsection*{Lemmas and their proofs}

\smallskip

\begin{lemma}\label{E}
Under Condition \ref{cond3}, the $n \times q$ random matrix $\bE = (\be_1, \ldots, \be_n )^T$ with rows $\be_i$  i.i.d. $\sim \mathcal{N}(\mathbf{0}, \mathbf{\Sigma})$ satisfies that for any $t > 0$, with probability at least $1- \exp(-\frac{t^2}{2})$,
\begin{align*}
\|\bE\|_{2} & \leq \gamma_u (\sqrt{n} + \sqrt{q} + t). 
\end{align*}

\label{lem:conc}
\end{lemma}

\begin{proof}[Proof of Lemma~\ref{E}]
First of all, for any $1 \leq i \leq n$, we have 
\begin{align*}
\mathrm{E}(\mathbf{\Sigma}^{-1/2} \be_i ) & = \mathbf{0}, \\
\mathrm{Cov} (\mathbf{\Sigma}^{-1/2} \be_i) & = \mathbf{\Sigma}^{-1/2} \mathbf{\Sigma} \mathbf{\Sigma}^{-1/2} = \mathbf{I},
\end{align*}
where $I$ denotes an identity matrix. Hence, $\mathbf{\Sigma}^{-1/2} \bE^T = \left(\mathbf{\Sigma}^{-1/2} \be_1, \ldots, \mathbf{\Sigma}^{-1/2} \be_n \right)$ is a $q \times n$ matrix with independent zero mean and unit variance entries. Standard random matrix theory \citep{rudelson2010non} gives that $\mathrm{E} \left(\| \mathbf{\Sigma}^{-1/2}\bE^T \|_2 \right) \leq \sqrt{n} + \sqrt{q}$. Further applying \citet[Lemma 3]{bunea2011optimal} gives
\begin{align*}
\mathbb{P} \left\{ \| \mathbf{\Sigma}^{-1/2}\bE^T \|_2 \geq \sqrt{n} + \sqrt{q} + t \right \} \leq \exp(-\frac{t^2}{2}) 
\end{align*}
for any $t > 0$. Under Condition \ref{cond3}, we have with probability at least $1- \exp(-\frac{t^2}{2})$, 
\begin{align*}
& \|\bE\|_2 = \|\bE^T\|_2 = \|\mathbf{\Sigma}^{1/2} \mathbf{\Sigma}^{-1/2} \bE^T\|_2 \leq \|\mathbf{\Sigma}^{1/2}\|_2 \cdot \|\mathbf{\Sigma}^{-1/2}\bE^T\|_2 \leq \gamma_u (\sqrt{n} + \sqrt{q} + t).
\end{align*}
It yields the result of Lemma \ref{E}.
\end{proof}


\begin{lemma}\label{Es} 
Suppose that the rows of the $n \times q$ random matrix $\bE = (\be_1, \ldots, \be_n )^T$ are independent sub-Gaussian random vectors with a common second moment matrix $\mathbf{\Sigma}^{\star}$, whose eigenvalues are bounded from above by some positive constant $\gamma_u^{\star}$. Then it holds that for any $t > 0$, with probability at least $1- 2\exp(-\frac{t^2}{2})$,
\begin{align*}
\|\bE\|_{2}  \leq \gamma^{\star} (\sqrt{n} + \sqrt{q} + t),
\end{align*}
where $\gamma^{\star}\geq2$ is a constant.
\end{lemma}

\begin{proof}[Proof of Lemma~\ref{Es}]
For some positive constants $c_1$, $C_1$ and any $t_1 > 0$,  in view of \citet[Theorem 5.39]{Eldar2012}, we have with probability at least $1-2\exp(-c_1t_1^2)$,
\begin{align*}
\|\frac{1}{n}\bE^{T}\bE-\mathbf{\Sigma}^{\star}\|_2\leq \max(\delta,\delta^2),
\end{align*}
where $\delta=C_1\sqrt{\frac{q}{n}}+\frac{t_1}{\sqrt{n}}$. Then it follows from the triangular inequality that
\begin{align}\label{yang1}
\|\frac{1}{n}\bE^{T}\bE-\bI\|_2\leq \|\frac{1}{n}\bE^{T}\bE-\mathbf{\Sigma}^{\star}\|_2+\|\mathbf{\Sigma}^{\star}-\bI\|_2.
\end{align}
Since the eigenvalues of $\mathbf{\Sigma}^{\star}$ are bounded from above by $\gamma_u^{\star}$, there exists some positive constant $\widetilde{\gamma}_u$ such that
\begin{align*}
\|\mathbf{\Sigma}^{\star}-\bI\|_2 \leq \widetilde{\gamma}_u. 
\end{align*}
This inequality along with (\ref{yang1}) yields
\begin{align*}
\|\frac{1}{n}\bE^{T}\bE-\bI\|_2\leq \|\frac{1}{n}\bE^{T}\bE-\mathbf{\Sigma}^{\star}\|_2+\|\mathbf{\Sigma}^{\star}-\bI\|_2\leq\max(\delta,\delta^2)+\gamma_u''\leq\max(\delta_1,\delta_1^2),
\end{align*}
where $\gamma_u''=\max(\widetilde{\gamma}_u,1/{\widetilde{\gamma}_u})\geq1$ and $\delta_1=\delta+\gamma_u''$.

Then with the aid of \citet[Lemma 5.36]{Eldar2012}, it follows from the above inequality that
\begin{align*}
\|\bE/\sqrt{n}\|_{2}  \leq 1+\delta_1= 1+ \gamma_u''+ C_1\sqrt{\frac{q}{n}}+\frac{t_1}{\sqrt{n}},
\end{align*}
which yields
\begin{align*}
\|\bE\|_{2}  \leq  (1+ \gamma_u'')\sqrt{n} + C_1\sqrt{q}+t_1.
\end{align*}
Then taking $t_1 = t/\sqrt{2c_1}$,  we have with probability at least $1- 2\exp(-\frac{t^2}{2})$,
\begin{align*}
\|\bE\|_{2}  \leq  (1+ \gamma_u'')\sqrt{n} + C_1\sqrt{q}+t/\sqrt{2c_1}\leq \gamma^{\star} (\sqrt{n} + \sqrt{q} + t),
\end{align*}
where $\gamma^{\star}=\max\{(1+ \gamma_u''),C_1,1/\sqrt{2c_1}\}\geq2$. It yields the result of Lemma~\ref{Es}.

\end{proof}

\end{document}